%% file: main.tex
\documentclass[10pt,journal]{IEEEtran}

\usepackage[table]{xcolor}

\usepackage{%
	amsmath,%
	accents,%
	amssymb,%
	amsthm,%
	graphicx,%
	array,%
	bchart,%
	caption,%
	cite,%
	diagbox,%
	epsfig,%
	latexsym,%
	multirow,%
	pgfplots,%
	sidecap,%
	subcaption,
	tikz,%
	varwidth,%
	wrapfig,%
	arydshln,%
	slashbox,%
	nccmath%
}

\usepackage[inline]{enumitem}
\usepackage[linesnumbered,ruled,vlined]{algorithm2e}

\sidecaptionvpos{figure}{c}

\usetikzlibrary{%
	arrows,%
	automata,%
	calc,%
	patterns,%
	plotmarks,%
	positioning,%
	shapes.geometric,%
	shapes,%
	snakes,%
}

\pgfplotsset{
    legend image with text/.style={
        legend image code/.code={%
            \node[anchor=center] at (0.3cm,0cm) {#1};
        }
    },
}

\newcommand{\be}{\begin{equation}}
\newcommand{\ee}{\end{equation}}
\newcommand{\bear}{\begin{eqnarray}}
\newcommand{\eear}{\end{eqnarray}}
\newcommand{\bears}{\begin{eqnarray*}}
\newcommand{\eears}{\end{eqnarray*}}
\newcommand{\bi}{\begin{itemize}}
\newcommand{\ei}{\end{itemize}}
\newcommand{\ben}{\begin{enumerate}}
\newcommand{\een}{\end{enumerate}}

\newcommand{\E}{\mathbb{E}}

\newcommand{\ie}{{i.e., }}

\newcommand{\eg}{{e.g., }}

\definecolor{ogreen}{rgb}{0,0.5,0}
\definecolor{magenta}{rgb}{1.0, 0.11, 0.81}
\definecolor{mulberry}{rgb}{0.77, 0.29, 0.55}
\definecolor{xgray}{rgb}{0.9, 0.9, 0.9}
\def \blue{\color{blue}}

\definecolor{byzantium}{rgb}{0, 0, 0.5}

\def \bes{\begin{equation*}}
\def \ees{\end{equation*}}
\def \bas{\begin{align*}}
\def \eas{\end{align*}}
\def \be{\begin{equation}}
\def \ee{\end{equation}}
\def \bbm{\begin{bmatrix}}
\def \ebm{\end{bmatrix}}

\def \mbx {\mathbf{x}}
\def \mbr {R}
\def \mba {A}
\definecolor{rawad}{rgb}{0.19, 0.55, 0.91}

\def \bes{\begin{equation*}}
\def \ees{\end{equation*}}
\def \bas{\begin{align*}}
\def \eas{\end{align*}}
\def \bbm{\begin{bmatrix}}
\def \ebm{\end{bmatrix}}

\def \m{{\fontfamily{concrete}\selectfont \text{M}}}

\newtheorem*{theorem*}{Theorem}
\newtheorem{theorem}{Theorem}
\newtheorem{lemma}[theorem]{Lemma}

\newtheorem{example}{Example}

\newtheorem{remark}{Remark}

\IEEEoverridecommandlockouts

\begin{document}

\newlength\figureheight
\newlength\figurewidth

\title{Private and Rateless Adaptive Coded Matrix-Vector Multiplication}

\author{
\IEEEauthorblockN{
        Rawad~Bitar, 
        Yuxuan~Xing, 
        Yasaman~Keshtkarjahromi, 
        Venkat Dasari, 
        Salim El Rouayheb, 
        and~Hulya~Seferoglu
        }


\thanks{
The preliminary results of this paper are presented in part at the SPIE Defense + Commercial Sensing, Baltimore, MD, April 2019 \cite{BXKDRS19}. }

\thanks{R. Bitar and S. El Rouayheb are with the Department of Electrical and Computer Engineering at Rutgers, the State University of New Jersey, Piscataway, New Jersey, 08854. E-mails: {\tt rawad.bitar@rutgers.edu}, {\tt salim.elrouayheb@rutgers.edu}.}
\thanks{Y. Xing, and H. Seferoglu are with the Department of Electrical and Computer Engineering, University of Illinois at Chicago, Chicago, IL, 60607. E-mails: {\tt hulya@uic.edu}, {\tt yxing7@uic.edu}.}
\thanks{Y. Keshtkarjahromi is with the Storage Research Group at the Seagate Technology, Shakopee, MN, 55379. E-mail: {\tt yasaman.keshtkarjahromi@seagate.com}.} 
\thanks{V. Dasari is with the US Army Research Lab, Aberdeen Proving Ground, MD.}

}

%
%

\maketitle

\begin{abstract}

Edge computing is emerging as a new  paradigm to allow  processing data near the edge of the network, where the data is typically generated and collected. This enables critical computations at the edge in applications such as Internet of Things (IoT), in which an  increasing number of devices (sensors, cameras, health monitoring devices, etc.)  collect  data that needs to be processed through computationally intensive algorithms with stringent reliability, security and latency constraints. 

Our key tool is the theory of coded computation, which advocates mixing data in computationally intensive tasks by employing erasure codes and offloading these tasks to other devices for computation. Coded computation is recently gaining interest, thanks to its higher reliability, smaller delay, and lower communication costs. In this paper, we develop a private and rateless adaptive coded computation (PRAC) algorithm for distributed matrix-vector multiplication by taking into account (i) the privacy requirements of IoT applications and devices, and (ii) the heterogeneous and time-varying resources of edge devices. We show that PRAC outperforms known secure coded computing methods when resources are heterogeneous. We provide theoretical guarantees on the performance of PRAC and its comparison to baselines. Moreover, we confirm our theoretical results through simulations and implementations on Android-based smartphones.

\end{abstract}

\section{Introduction}

Edge computing is emerging as a new  paradigm to allow  processing data near the edge of the network, where the data is typically generated and collected. This enables computation at the edge in applications such as Internet of Things (IoT), in which an  increasing number of devices (sensors, cameras, health monitoring devices, etc.)  collect  data that needs to be processed through computationally intensive algorithms with stringent reliability, security and latency constraints. 

One of the promising solutions to handle computationally intensive tasks is computation offloading, which advocates offloading tasks to remote servers or cloud. Yet, offloading tasks to remote servers or cloud could be luxury that cannot be afforded by most of the edge applications, where connectivity to remote servers can be lost or compromised, which makes edge computing crucial.

Edge computing advocates that computationally intensive tasks in a device (master) could be offloaded to other edge or end devices (workers) in close proximity. 
%
However, offloading tasks to other devices leaves the IoT and the applications it is supporting at the complete mercy of an attacker. Furthermore, exploiting the potential of edge computing is challenging mainly due to the heterogeneous and time-varying nature of the devices at the edge. Thus, our goal is to develop a private, dynamic, adaptive, and heterogeneity-aware cooperative computation framework that provides both privacy and computation efficiency guarantees. Note that the application of this work can be extended to cloud computing at remote data-centers. However, we focus on edge computing as heterogeneity and time-varying resources are more prevalent at the edge as compared to data-centers.


Our key tool is the theory of coded computation, which advocates mixing data in computationally intensive tasks by employing erasure codes and offloading these tasks to other devices for computation \cite{KS18,BPR17,li2016unified,dutta2017coded,yang2017computing,halbawi2017improving,yu2017polynomial,DCG16,tandon2017gradient,li2016fundamental,lee2018speeding,karakus2017straggler,aktas2017effective}. The following canonical example demonstrates the effectiveness of coded computation. 

\begin{example} \label{ex:ex1}
Consider the setup where a master device  wishes to offload a task to 3 workers. The master has a large data  matrix $A$  and wants to compute matrix vector product $A\mathbf{x}$. 
The master device divides the matrix $A$ row-wise equally into two smaller matrices $A_1$ and $A_2$, which are then encoded using a $(3,2)$ Maximum Distance Separable (MDS) code\footnote{An $(n,k)$ MDS code divides the master's data into $k$ chunks and encodes it into $n$ chunks ($n>k$) such that any $k$ chunks out of $n$ are sufficient to recover the original data. } to give  $B_1=A_1$, $B_2=A_2$ and $B_3=A_1+A_2$, and sends each to a different worker. Also, the master device sends $\mathbf{x}$ to workers and ask them to compute $B_i\mathbf{x}$, $i \in \{1,2,3\}$. When the master receives the computed values (\ie $B_i\mathbf{x}$) from at least two out of three workers, it can decode its desired task, which is the computation of $A\mathbf{x}$. The power of coded computations is that it makes $B_3=A_1+A_2$ act as a ``joker" redundant task that can replace any of the other two tasks if they end up straggling or failing.
\hfill $\Box$
\end{example}

The above example demonstrates the benefit of coding for edge computing. However, the very nature of task offloading from a master to worker devices makes the computation framework vulnerable to attacks. One of the attacks, which is also the focus of this work, is {\em eavesdropper adversary}, where one or more of workers can behave as an eavesdropper and can spy on the coded data sent to these devices for computations.\footnote{Note that this work focuses specifically on {\em eavesdropper adversary} although there are other types of attacks; for example {\em Byzantine adversary}, which is out of scope of this work.} 
For example, $B_3=A_1+A_2$ in Example~\ref{ex:ex1} can be processed and spied by worker 3. Even though $A_1+A_2$ is coded, the attacker can infer some information from this coded task. Privacy against eavesdropper attacks is extremely important in edge computing \cite{shirazi2017extended,abbas2017mobile,roman2018mobile}. Thus, it is crucial to develop a private coded computation mechanism against eavesdropper adversary who can gain access to offloaded tasks.

In this paper, we develop a private and rateless adaptive coded computation (PRAC) mechanism. PRAC is (i) private as it is secure against eavesdropper adversary, (ii) rateless, because it uses Fountain codes \cite{LT, Raptor, Fountain} instead of Maximum Distance Separable (MDS) codes \cite{lin1983error,macwilliams1977theory}, and (iii) adaptive as the master device offloads tasks to workers by taking into account their heterogeneous and time-varying resources. Next, we illustrate the main idea of PRAC through an illustrative example.

\begin{table}[t!]
\caption{Example PRAC operation in heterogeneous and time-varying setup.}
\centering
\normalsize{
\begin{tabular}{c|c|c|c}
\underline{Time} & \underline{Worker 1}& \underline{Worker 2} & \underline{Worker 3} \\
$1$ & $R_1$ & $ {\blue \bf A_1+A_3} + R_1$ & $ {\blue \bf A_3} + R_1$  \\
$2$ & \cellcolor{gray!25} & $R_2$ &\cellcolor{gray!25} \\
$3$ & ${\blue \bf A_2 + A_3 } + R_2$ &\cellcolor{gray!25} &\cellcolor{gray!25} \\
$4$ &\cellcolor{gray!25} & \cellcolor{gray!25}&$ {\blue \bf A_2} + R_2$
\end{tabular}
}
\vspace{-0.2cm}
\label{tab:ex1}
\end{table}

\begin{example} \label{ex:ex2}
We consider the same setup in Example~\ref{ex:ex1}, where a master device offloads a task to 3 workers. The master has a large data  matrix $A$ and wants to compute matrix vector product $A\mathbf{x}$. The master device divides matrix $A$ row-wise into $3$ sub-matrices $A_1$, $A_2$, $A_3$; and encodes these matrices using a Fountain code\footnote{Fountain codes are desirable here for two properties: (i) they provide a fluid abstraction of the coded packets so the master can always decode with high probability as long as it collects enough packets; (ii) They have low decoding complexity.} \cite{LT, Raptor, Fountain}. An example set of coded packets is $A_2$, $A_3$, $A_1+A_3$, and $A_2+A_3$. However, prior to sending a coded packet to a worker, the master generates a random key matrix $R$ with the same dimensions as $A_i$ and with entries drawn uniformly from the same alphabet as the entries of $A$. The key matrix is added to the coded packets to provide privacy as shown in Table~\ref{tab:ex1}. In particular, a key matrix $R_1$ is created at the start of time slot $1$, combined with $A_1+A_3$ and $A_3$, and transmitted to workers~$2$ and $3$, respectively. $R_1$ is also transmitted to worker~$1$ in order to obtain $R_1\mathbf{x}$ that will help the master in the decoding process. The computation of $(A_1+A_3+R_1)\mathbf{x}$ is completed at the end of time slot $1$. Thus, at that time slot the master generates a new matrix, $R_2$, and sends it to worker~$2$. At the end of time slot $2$, worker~$1$ finishes its computation, therefore the master adds $R_2$ to $A_2+A_3$ and sends it to worker~$1$. A similar process is repeated at the end of time slot $3$. Now the master waits for worker~$2$ to return $R_2\mathbf{x}$ and for any other worker to return its uncompleted task in order to decode $A\mathbf{x}$. Thanks to using key matrices $R_1$ and $R_2$, and assuming that workers do not collude, privacy is guaranteed. On a high level, privacy is guaranteed because the observation of the workers is statistically independent from the data $A$.
%
\hfill $\Box$
\end{example}

This example shows that PRAC can take advantage of coding for computation, and provide privacy.

{\em Contributions.} We design PRAC for heterogeneous and time-varying private coded computing with colluding workers. In particular, PRAC codes sub-tasks using Fountain codes, and determines how many coded packets and keys each worker should compute dynamically over time. We provide theoretical analysis of PRAC and show that it (i) guarantees privacy conditions, (ii) uses minimum number of keys to satisfy privacy requirements, and (iii) maintains the desired rateless property of non-private Fountain codes. Furthermore, we provide a closed form task completion delay analysis of PRAC. Finally, we evaluate the performance of PRAC via simulations as well as in a testbed consisting of real Android-based smartphones as compared to baselines. 

The use of Fountain codes in encoding the sub-tasks provides PRAC flexibility on the number of stragglers and on the computing capacity of workers, reflected by the number of sub-tasks assigned to each worker. In contrast, existing solutions for secure coded computing require the master to set a threshold on the number of stragglers that it can tolerate and pre-assign the sub-tasks to the workers based on this threshold. 


{\em Organization.} The structure of the rest of this paper is as follows. We start with presenting the system model in Section~\ref{sec:system}. Section~\ref{sec:PRAC} presents the design of private and rateless adaptive coded computation (PRAC). We characterize and analyze PRAC in Section~\ref{sec:theo}. We present evaluation results in section~\ref{sec:exp}. Section~\ref{sec:related} presents related work. Section~\ref{sec:conc} concludes the paper.





\section{\label{sec:system} System Model}

{\em Setup.} We consider a master/workers setup at the edge of the network, where the master device $\m$ offloads its computationally intensive tasks to workers $w_i,\ i \in \mathcal{N},$ (where $|\mathcal{N}| = n$) via device-to-device (D2D) links such as Wi-Fi Direct and/or Bluetooth. The master device divides a task into smaller sub-tasks, and offloads them to workers that process these sub-tasks in parallel.


{\em Task Model.}
We focus on the computation of linear functions, \ie matrix-vector multiplication. We suppose the master wants to compute the matrix vector product $A\mathbf{x}$, where $A \in \mathbb{F}_q^{m\times \ell}$ can be thought of as the data matrix and $\mathbf{x}\in \mathbb{F}_q^{\ell}$ can be thought of as an attribute vector. We assume that the entries of $A$ and $\mathbf{x}$ are drawn independently and uniformly at random\footnote{We abuse notation and denote both the random matrix representing the data and its realization by $A$. We do the same for $\mathbf{x}$.} from $\mathbb{F}_q$. The motivation stems from machine learning applications where computing linear functions is a building block of several iterative algorithms \cite{seber2012linear,suykens1999least}. For instance, the main computation of a gradient descent algorithm with squared error loss function is \begin{equation}\label{eq:gradup}
\mathbf{x}^{+} = \mathbf{x} - \alpha A^T(A\mathbf{x}-\mathbf{y}),
\end{equation}
where $\mathbf{x}$ is the value of the attribute vector at a given iteration, $\mathbf{x}^{+}$ is the updated value of $\mathbf{x}$ at this iteration and the learning rate $\alpha$ is a parameter of the algorithm. Equation~\eqref{eq:gradup} consists of computing two linear functions $A\mathbf{x}$ and $A^T\mathbf{w}\triangleq A^T(A\mathbf{x}-\mathbf{y})$.




{\em Worker and Attack Model.} 
The workers incur random delays while executing the task assigned to them by the master device. The workers have different computation and communication specifications resulting in a heterogeneous environment which includes workers that are significantly slower than others, known as stragglers. Moreover, the workers cannot be trusted with the master's data. We consider an {\em eavesdropper adversary} in this paper, where one or more of workers can be eavesdroppers and can spy on the coded data sent to these devices for computations. We assume that up to $z$, $z<n$, workers can collude, \ie $z$ workers can share the data they received from the master in order to obtain information about $A$. The parameter $z$ can be chosen based on the desired privacy level; a larger $z$ means a higher privacy level and vice versa. One would want to set $z$ to the largest possible value for maximum, $z=n-1$ security purposes. However, this has the drawback of increasing the complexity and the runtime of the algorithm. In our setup we assume that $z$ is a fixed and given system parameter.


{\em Coding \& Secret Keys.} The matrix $A$ can be divided into $b$ row blocks (we assume that $b$ divides $m$, otherwise all-zero rows can be added to the matrix to satisfy this property) denoted by $A_i$, $i=1,\dots, b$. The master applies Fountain coding \cite{LT, Raptor, Fountain} across row blocks to create information packets $\nu_j \triangleq \sum_{i=1}^m c_{i,j}A_i$, $j=1,2,\dots,$ where the $c_{i,j}\in\{0,1\}$. 
Note that an information packet is a matrix of dimension $m/b \times \ell$, i.e., $\nu_j \in \mathbb{F}_q^{m/b\times \ell}$. 
Such rateless coding is compatible with our goal to create adaptive coded cooperation computation framework.
In order to maintain privacy of the data, the master device generates random matrices $R_i$ of dimension $m/b\times \ell$ called \textit{keys}. The entries of the $\mbr_i$ matrices are drawn uniformly at random from the same field as the entries of $A$. Each information packet $\nu_j$ is {\em padded} with a linear combination of $z$ keys $f_j( R_{i,1},\dots,R_{i,z})$ to create a secure packet $s_j \in \mathbb{F}_q^{m/b \times \ell}$ defined as $s_j\triangleq \nu_j + f_j( R_{i,1},\dots,R_{i,z})$.

The master device sends $\mathbf{x}$ to all workers, then it sends the keys and the $s_j$'s to the workers according to our PRAC scheme described later. Each worker multiplies the received packet by $\mathbf{x}$ and sends the result back to the master. Since the encoding is rateless, the master keeps sending packets to the workers until it can decode $A\mathbf{x}$. The master then sends a stop message to all the workers.

{\em Privacy Conditions.}
Our primary requirement is that any collection of $z$ (or less) workers will not be able to obtain any information about $A$, in an information theoretic sense.

In particular, let $P_i$, $i=1\dots,n$, denote the collection of packets sent to worker $w_i$. For any set $\mathcal{B}\subseteq \{1,\dots,n\}$, let $P_\mathcal{B}\triangleq\{P_i, i\in\mathcal{B}\}$ denote the collection of packets given to worker $w_i$ for all $i\in \mathcal{B}$. The privacy requirement\footnote{In some cases the vector $\mathbf{x}$ may contain information about $A$ and therefore must not be revealed to the workers. We explain in Appendix~A how to generalize our scheme to account for such cases.} can be expressed as
\begin{align}\label{eq:priv}
 H(A|P_\mathcal{Z})&=H(A),\quad  \forall \mathcal{Z}\subseteq \{1,\dots,n\} \text{ s.t. } |\mathcal{Z}|\leq z.
 \end{align}

$H(A)$ denotes the entropy, or uncertainty, about $A$ and $ H(A|P_\mathcal{Z})$ denotes the uncertainty about $A$ after observing $P_\mathcal{Z}$.

{\em Delay Model.} 
Each packet transmitted from the master to a worker $w_i,\ i=1,2,...,n,$ experiences the following delays: (i) transmission delay for sending the packet from the master to the worker, (ii) computation delay for computing the multiplication of the packet by the vector $\mathbf{x}$, and (iii) transmission delay for sending the computed packet from the worker $w_i$ back to the master. We denote by $\beta_{t,i}$ the computation time of the $t^\text{th}$ packet at worker $w_i$ and $RTT_i$ denotes the {average} round-trip time spent to send and receive a packet from worker $w_i$. The time spent by the master is equal to the time taken by the $(z+1)^\text{st}$ fastest worker to finish its assigned tasks.

\begin{table}
\normalsize
    \caption{Summary of notations.}
    \label{tab:my_label}
    \centering
    \begin{tabular}{c|c}
       Symbol  & Meaning \\ \hline
       $\m$ & master\\
       $w_i$ & worker $i$ \\
       $n$  & number of workers \\
       $A$ & $m\times \ell$ data matrix\\
       $\mathbf{x}$ & $\ell \times 1$ attribute vector\\
       $z$ & number of colluding workers\\
       $m$ & number of rows in $A$\\
       $\varepsilon$ & overhead of Fountain codes\\
        $A_i$ & $i^\text{th}$ row block of data matrix $A$\\
       $R$ & random matrix\\
       $RTT_i$ & {average} round trip time to send and receive packet $p_i$\\
       $\beta_{t,i}$ & computation time of the $t^\text{th}$ packet at $w_i$\\
       $\nu$ & Fountain coded packet of $A_i$'s\\
       $s$ & secure Fountain coded packet\\
       $T_i$ & time to compute a packet at $w_i$\\
       $T_{(d)}$ & $d^\text{th}$ order statistic of $T_i$'s\\
       $T$ & time spent by $\m$ to decode $A\mathbf{x}$\\

    \end{tabular}

\end{table}

\begin{table*}[t]
\caption{Depiction of PRAC in the presence of stragglers. The master keeps generating packets using Fountain codes until it can decode $\mba \mbx$. The master estimates the average task completion time of each worker and sends a new packet to avoid idle time. Each new packet sent to a worker must be secured with a new random key. The master can decode $\mba_1\mathbf{x}, \dots, \mba_6 \mathbf{x}$ after receiving all the packets not having $\mbr_{4,1}$ or $\mbr_{4,2}$ in them.}
\label{fig:colfig2}
\centering
\resizebox{0.95\textwidth}{!}{
\begin{tikzpicture}
\node at (0,0) {
\begin{tabular}{c|c|c|c|c|c}
\underline{Time} & \underline{Worker $1$}& \underline{Worker $2$} & \underline{Worker $3$} & \underline{Worker $4$} \\
1 &$\mbr_{1,1}$ & $\mbr_{1,2}$ & $ {\blue \bf \mba_4} + \mbr_{1,1} + \mbr_{1,2}$ & ${\blue \bf \mba_3 + \mba_4 + \mba_6}+ \mbr_{1,1} + 2 \mbr_{1,2}$ \\
2 & \cellcolor{gray!25} & \cellcolor{gray!25} &\cellcolor{gray!25}&$\mbr_{2,1}$\\
3 & $\mbr_{2,2}$ & \cellcolor{gray!25}&\cellcolor{gray!25}&\cellcolor{gray!25}\\
4 &\cellcolor{gray!25} &  $ {\blue \bf \mba_3} + \mbr_{2,1} + \mbr_{2,2}$  & $ {\blue \bf \mba_4 + \mba_5 } + \mbr_{2,1} + 2\mbr_{2,2}$ & \cellcolor{gray!25} \\
5& \cellcolor{gray!25}& $ \mbr_{3,1}$ & \cellcolor{gray!25} & \cellcolor{gray!25}\\
6&  ${\blue \bf \mba_2 } +\mbr_{3,1} + \mbr_{3,2}$ &\cellcolor{gray!25} &\cellcolor{gray!25}&$ \mbr_{3,2}$\\
7&  \cellcolor{gray!25}& $ \mbr_{4,1}$ & ${\blue \bf \mba_1 } +\mbr_{3,1} +2 \mbr_{3,2}$&\cellcolor{gray!25}\\
8&  $ \mbr_{4,2}$& \cellcolor{gray!25} & \cellcolor{gray!25} & ${\blue \bf \mba_2+\mba_3 } +\mbr_{4,1} + \mbr_{4,2}$\\
\end{tabular}};
\end{tikzpicture}
}

\end{table*}

\section{Design of PRAC}\label{sec:PRAC}

\subsection{Overview} 
We present the detailed explanation of PRAC. Let $p_{t,i}\in\mathbb{F}_q^{m/b\times \ell}$ be the $t^\text{th}$ packet sent to worker $w_i$. This packet can be either a key or a secure packet. For each value of $t$, the master sends $z$ keys denoted by $R_{t,1},\dots, R_{t,z}$ to $z$ different workers and up to $n-z$ secure packets $s_{t,1},\dots,s_{t,n-z}$ to the remaining workers. The master needs the results of $b+\epsilon$ information packets, \ie $\nu_{t,i}\mathbf{x}$, to decode the final result $A\mathbf{x}$, where $\epsilon$ is the overhead required by Fountain coding\footnote{The overhead required by Fountain coding is typically as low as $5\%$ \cite{Fountain}, \ie $\epsilon=0.05b$}. To obtain the results of $b+\epsilon$ information packets, the master needs the results of $b+\epsilon$ secure packets, $s_{t,i}\mathbf{x} = (\nu_{i,j}+f_j(R_{t,i},\dots,R_{t,z}))\mathbf{x}$, together with all the corresponding\footnote{Recall that $f_j( R_{t,1},\dots,R_{t,z})$ is a linear function, thus it is easy to extract $(R_{t,i})\mathbf{x}, i=1,...,z$, from $(f_j(R_{t,1},\dots,R_{t,z}))\mathbf{x}$.} $R_{t,i}\mathbf{x}, i=1,\dots,z$. Therefore, only the results of the $s_{t,i}\mathbf{x}$ for which all the computed keys $R_{t,i}\mathbf{x}, i=1,...,z,$ are received by the master can account for the total of $b+\epsilon$ information packets.

\subsection{\label{sec:DRA} Dynamic Rate Adaptation} 

The dynamic rate adaptation part of PRAC is based on \cite{KS18}. In particular, 
the master offloads coded packets gradually to workers and receives two ACKs for each transmitted packet; one confirming the receipt of the packet by the worker, and the second one (piggybacked to the computed packet) showing that the packet is computed by the worker. Then, based on the frequency of the received ACKs, the master decides to transmit more/less coded packets to that worker. In particular, each packet $p_{t,i}$ is transmitted to each worker $w_i$ before or right after the computed packet $p_{t-1,i}\mathbf{x}$ is received at the master. For this purpose, the average per packet computing time $\E[\beta_{t,i}]$ is calculated for each worker $w_i$ dynamically based on the previously received ACKs. Each packet $p_{t,i}$ is transmitted after waiting $\E[\beta_{t,i}]$ from the time $p_{t-1,i}$ is sent or right after packet $p_{t-1,i}\mathbf{x}$ is received at the master, thus reducing the idle time at the workers. This policy is shown to approach the optimal task completion delay and maximizes the workers' efficiency and is shown to improve task completion time significantly compared with the literature \cite{KS18}.

\subsection{Coding}\label{sec:Cod}

We explain the coding scheme used in PRAC. We start with an example to build an intuition and illustrate the scheme before going into details.


\begin{example}
Assume there are $n=4$ workers out of which any $z = 2$ can collude. Let $A$ and $\mathbf{x}$ be the data owned by the master and the vector to be multiplied by $A$, respectively. The master sends $\mathbf{x}$ to all the workers. For the sake of simplicity, assume $A$ can be divided into $b=6$ row blocks, \ie $A=\bbm A_1^T & A_2^T & \dots & A_6^T \ebm^T$. The master encodes the $A_i$'s using Fountain code. We denote by \emph{round} the event when the master sends a new packet to a worker. For example, we say that worker~$1$ is at round~$3$ if it has received $3$ packets so far. For every round $t$, the master generates $z=2$ random matrices $R_{t,1},\ \mbr_{t,2}$ (with the same size as $A_1$) and encodes them using an $(n,z)=(4,2)$ systematic maximum distance separable (MDS) code by multiplying $\mbr_{t,1},\ \mbr_{t,2}$ by a generator matrix $G$ as follows
\be
G \bbm \mbr_{t,1}\\ \mbr_{t,2}\ebm \triangleq 
\bbm 
1 & 0 \\
0 & 1 \\
1 & 1\\
1 & 2\\
\ebm
\bbm \mbr_{t,1}\\ \mbr_{t,2}\ebm.
\ee
This results in the encoded matrices of $\mbr_{t,1}$, $\mbr_{t,2}$, $\mbr_{t,1}+\mbr_{t,2}$, and $\mbr_{t,1}+2\mbr_{t,2}$. %
Now let us assume that workers can be stragglers. At the beginning, the master initializes all the workers at round~$1$. Afterwards, when a worker $w_i$ finishes its task, the master checks how many packets this worker has received so far and how many other workers are at this round. If this worker $w_i$ is the first or second to be at round $t$, the master generates $\mbr_{t,1}$ or $\mbr_{t,2}$, respectively, and sends it to $w_i$. Otherwise, if $w_i$ is the $j^\text{th}$ worker ($j>2$) to be at round $t$, the master multiplies $\bbm \mbr_{t,1}&\mbr_{t,2}\ebm^T$ by the $j^{\text{th}}$ row of $G$, adds it to a generated Fountain coded packet, and sends it to $w_i$. The master keeps sending packets to the workers until it can decode $\mba \mbx$. We illustrate the idea in Table~\ref{fig:colfig2}.

\end{example}

We now explain the details of PRAC in the presence of $z$ colluding workers. 
\begin{enumerate}
\item {\em Initialization:} The master divides $A$ into $b$ row blocks $\mba_1,\dots,\mba_b$ and sends the vector $\mathbf{x}$ to the workers. Let $G \in \mathbb{F}_q^{n\times z}$, $q>n$, be the generator matrix of an $(n,z)$ systematic MDS code. For example one may use systematic Reed-Solomon codes that use Vandermonde matrix as generator matrix, see for example \cite{lacan2009reed}. The master generates $z$  random matrices $\mbr_{1,1},\dots, \mbr_{1,z}$ and encodes them using $G$. Each coded key can be denoted by $\mathbf{g}_i\mathcal{R}$ where $\mathbf{g}_i$ is the $i^\text{th}$ row of $G$ and $\mathcal{R}\triangleq \bbm \mbr_{1,1}^T&\dots& \mbr_{1,z}^T\ebm^T$. The master sends the $z$ keys $\mbr_{1,1},\dots, \mbr_{1,z}$ to the first $z$ workers, generates $n-z$ Fountain coded packets of the $\mba_i$'s, adds to each packet an encoded random key $\mathbf{g}_i\mathcal{R}$, $i=z+1,\dots n$, and sends them to the remaining $n-z$ workers.

\item {\em Encoding and adaptivity:} When the master wants to send a new packet to a worker (noting that a packet $p_{t,i}$ is transmitted to worker $w_i$ before, or right after, the computed packet $p_{t-1,i}\mathbf{x}$ is received at the master according to the strategy described in Section~\ref{sec:DRA}), it checks at which round this worker is, \ie how many packets this worker has received so far, and checks how many other workers are at least at this round. Assume worker $w_i$ is at round $t$ and $j-1$ other workers are at least at this round. If $j\leq z$, the master generates and sends $\mbr_{t,j}$ to the worker. However, if $j>z$ the master generates a Fountain coded packet of the $\mba_i$'s (\eg $A_1+A_2$), adds to it $\mathbf{g}_j\mathcal{R}$ and sends the packet ($A_1+A_2 + \mathbf{g}_j\mathcal{R}$) to the worker. Each worker computes the multiplication of the received packet by the vector $\mathbf{x}$ and sends the result to the master.

\item {\em Decoding and speed:} Let $\tau_i$ denote the number of packets sent to worker $i$. We define $\tau_{max}\triangleq \max_i \tau_i$ such that at the end of the process the master has $\mbr_{t,i} \mathbf{x}$ for all $t=1,\dots,\tau_{max}$ and all $i=1,\dots,z$. The master can therefore subtract $\mbr_{t,i}$, $t=1,\dots,\tau_{max}$ and $i=1,\dots,z$, from all received secure information packets, and thus can decode the $\mba_i$'s using the Fountain code decoding process. The number of secure packets that can be used to decode the $\mba_i$'s is dictated by the $(z+1)^\text{st}$ fastest worker, \ie the master can only use the results of secure information packets computed at a given round if at least $z+1$ workers have completed that round. If for example the $z$ fastest workers have completed round $100$ and the $(z+1)^\text{st}$ fastest worker has completed round $20$, the master can only use the packets belonging to the first $20$ rounds. The reason is that the master needs all the keys corresponding to a given round in order to use the secure information packet for decoding. In~Lemma~\ref{lemma:z+1} we prove that this scheme is optimal, \ie in private coded computing the master cannot use the packets computed at rounds finished by less than $z+1$ workers irrespective of the coding scheme.
\end{enumerate}


\section{Performance Analysis of PRAC}\label{sec:theo}


\subsection{Privacy}
In this section, we provide theoretical analysis of PRAC by particularly focusing on its privacy properties. 
\begin{theorem}\label{thm:pracPrivacy}
PRAC is a rateless real-time adaptive coded computing scheme that allows a master device to run distributed linear computation on private data $A$ via $n$ workers while satisfying the privacy constraint given in~\eqref{eq:priv} for a given $z<n$. 
\end{theorem}
\begin{proof}
Since the random keys are generated independently at each round, it is sufficient to study the privacy of the data on one round and the privacy generalizes to the whole algorithm. We show that for any subset $Z\subset\{1,\dots,n\}, |{Z}| = z$, the collection of packets $p_{Z}\triangleq \{p_{t,i}, i\in Z\}$ sent at round $t$ reveals no information about the data $A$ as given in \eqref{eq:priv}, \ie $H(A) = H(A|p_Z)$. Let $K$ denote the random variable representing all the keys generated at round $t$, then it is enough to show that $H(K|A,p_{Z} ) = 0$ as detailed in Appendix~\ref{appen:proofPrivacy}. Therefore, we need to show that given $A$ as side information, any $z$ workers can decode the random keys $R_{t,1}, \dots, R_{t,z}$. Without loss of generality assume the workers are ordered from fastest to slowest, \ie worker $w_1$ is the fastest at the considered round $t$. Since the master sends $z$ random keys to the fastest $z$ workers, then $p_{t,i}=R_{t,i}, i=1,\dots,z$. The remaining $n-z$ packets are secure information packets sent to the remaining $n-z$ workers, \ie $p_{t,i}=s_{t,i}= \nu_{t,i} + f(R_{t,1}, \dots, R_{t,z})$, where $\nu_{t,i}$ is a linear combination of row blocks of $A$ and $f(R_{t,1}, \dots, R_{t,z})$ is a linear combination of the random keys generated at round $t$. Given the data $A$ as side information, any collection of $z$ packets can be expressed as $z$ codewords of the $(n,z)$ MDS code encoding the random keys. Thus, given the matrix $A$, any collection of $z$ packets is enough to decode all the keys and $H(K|S,p_{Z} ) = 0$ which concludes the proof.
\end {proof}

\begin{remark}
PRAC requires the master to wait for the $(z+1)^\text{st}$ fastest worker in order to be able to decode $A\mathbf{x}$. We show in Lemma~\ref{lemma:z+1} that this limitation is a byproduct of all private coded computing schemes.
\end{remark}

\begin{remark}
PRAC uses the minimum number of keys required to guarantee the privacy constraints. At each round PRAC uses exactly $z$ random keys which is the minimum amount of required keys (c.f. Equation~\eqref{eq:keys} in Appendix~\ref{appen:proofPrivacy}).\end{remark}

\begin{lemma}\label{lemma:z+1}
Any private coded computing scheme for distributed linear computation limits the master to the speed of the $(z+1)^\text{st}$ fastest worker.
\end{lemma}
\begin{proof}
The proof of Lemma~\ref{lemma:z+1} is provided in Appendix~\ref{appen:proofLm}.
\end {proof}

\subsection{Task Completion Delay}
In this section, we characterize the task completion delay of PRAC and compare it with Staircase codes \cite{BPR17}, which are secure against eavesdropping attacks in a coded computation setup with homogeneous resources. First, we start with task completion delay characterization of PRAC.

\begin{theorem}\label{thm:pracCompTime}
Let $b$ be the number of row blocks in $A$, let $\beta_{t,i}$ denote the computation time of the $t^\text{th}$ packet at worker $w_i$ and let $RTT_i$ denote the {average} round-trip time spent to send and receive a packet from worker $i$. The task completion time of PRAC is approximated as
\begin{align}\label{eq:ETp}
T_{PRAC} \approx &\max_{i\in \{1,\dots,n\}} \{ RTT_i\} + \frac{b+\epsilon}{\sum_{i=z+1}^n 1/\E[\beta_{t,i}]},\\
\approx &\frac{b+\epsilon}{\sum_{i=z+1}^n 1/\E[\beta_{t,i}]}, \label{eq:T_{PRAC}}
\end{align}
where $w_i$'s are ordered indices of the workers from fastest to slowest, \ie $w_1 = \arg \min_{i} \mathbb{E}[\beta_{t,i}]$.
\end{theorem}
\begin{proof}
The proof of Theorem~\ref{thm:pracCompTime} is provided in Appendix~\ref{appen:proofTh3}. 
\end{proof}




Now that we characterized the task completion delay of PRAC, we can compare it with the state-of-the-art. Secure coded computing schemes that exist in the literature usually use static task allocation, where tasks are assigned to workers a priori. 
The most recent work in the area is Staircase codes, which is 
shown to outperform all existing schemes that use threshold secret sharing \cite{BPR17}. However, Staircase codes are static; they allocate fixed amount of tasks to workers a priori. Thus, Staircase codes cannot leverage the heterogeneity of the system, neither can it adapt to a system that is changing in time. On the other hand, our solution PRAC adaptively offloads tasks to workers by taking into account the heterogeneity and time-varying nature of resources at workers. Therefore, we restrict our focus on comparing PRAC to Staircase codes. 

Staircase codes assigns a task of size $b/(k-z)$ row blocks to each worker.\footnote{Note that in addition to $n$ and $z$, all threshold secret sharing based schemes require a parameter $k, \ z<k<n,$ which is the minimum number of non stragglers that the master has to wait for before decoding $A\mathbf{x}$.} Let $T_i$ be the time spent at worker $i$ to compute the whole assigned task. Denote by $T_{(i)}$ the $i^{th}$ order statistic of the $T_i$'s and by $T_{\text{SC}}(n,k,z)$ the task completion time, \ie  time the master waits until it can decode $A\mathbf{x}$, when using Staircase codes. In order to decode $A\mathbf{x}$ the master needs to receive a fraction equal to $(k-z)/(d-z)$ of the task assigned to each worker from any $d$ workers where $k\leq d \leq n$. The task completion time of the master can then be expressed as \cite{BPR17}
\begin{equation}\label{eq:sub}
T_{\text{SC}}(n,k,z)=\min_{d\in\{k,\dots,n\}}\left\{\dfrac{k-z}{d-z}T_{(d)}\right\}.
\end{equation}

\begin{theorem}
\label{th:gap_PRAC_staircase}
The gap between the completion time of PRAC and coded computation using staircase codes is lower bounded by:
\begin{align}
\mathbb{E}[T_{\text{SC}}]-\mathbb{E}\left[T_{PRAC}\right] \geq \frac{bx-\epsilon y}{y(x+y)},
\end{align}
where $x=\dfrac{n-d^*}{E[\beta_{t,n}]}$, $y=\dfrac{d^*-z}{E[\beta_{t,d^*}]}$ and $d^*$ is the value of $d$ that minimizes equation~\eqref{eq:sub}.
\end{theorem}
\begin{proof}
The proof of Theorem~\ref{th:gap_PRAC_staircase} is provided in Appendix~\ref{appen:proofTh4}.
\end {proof}


\begin{figure*}
\centering
 \setlength\figureheight{0.35\textwidth}
  \setlength\figurewidth{0.4\textwidth}
\begin{minipage}[b]{0.32\textwidth}
  \centering
\captionsetup{subtype,width=\textwidth}
\resizebox{0.85\textwidth}{!}{
\input{S2removeslowest_Rvariable_NumWork50_NumColl13_iter100}
}
\caption{Scenario~1 with the fastest $13$ workers as eavesdropper for GC3P 1 and the slowest workers as eavesdropper for GC3P 2.}
\label{fig:thm1}
\end{minipage}\hfill %
\begin{minipage}[b]{0.32\textwidth}
\centering
\captionsetup{subtype,width=\textwidth}
\resizebox{0.9\textwidth}{!}{
\input{S2removefastest_Rvariable_NumWork50_NumColl13_iter100}
}
\caption{Scenario~2 with $13$ workers picked at random to be eavesdroppers.}
\label{fig:thm2}
\end{minipage}\hfill %
\begin{minipage}[b]{0.32\textwidth}
  \centering
\captionsetup{subtype,width=\textwidth}
\resizebox{0.9\textwidth}{!}{
\input{S4nccpRvariable_NumWork50_NumColl13_iter100}
}
\caption{Scenario~3 with $13$ workers picked at random to be eavesdroppers.}
\label{fig:thm3}
\end{minipage}
\caption{Comparison between PRAC and the baselines Staircase codes, GC3P, and C3P in different scenarios with $n=50$ workers and $z=13$ colluding eavesdroppers for different values of the number of rows $m$. For each value of $m$ we run $100$ experiments and average the results. When the eavesdropper are chosen to be the fastest workers, PRAC has very similar performance to GC3P. When the eavesdroppers are picked randomly, the performance of PRAC becomes closer to this of GC3P when the non adversarial workers are more heterogeneous.}
\label{fig:remove}
\end{figure*}


Theorem \ref{th:gap_PRAC_staircase} shows that the lower bound on the gap between secure coded computation using staircase codes and PRAC is in the order of number of row blocks of $A$. Hence, the gap between secure coded computation using Staircase codes and PRAC is linearly increasing with the number of row blocks of $A$. Note that, $\epsilon$, the required overhead by fountain coding used in PRAC, becomes negligible as $b$ increases.

Thus, PRAC outperforms secure coded computation using Staircase codes in heterogeneous systems. The more heterogeneous the workers are, the more improvement is obtained by using PRAC. However, Staircase codes can slightly outperform PRAC in the case where the slowest $n-z$ workers are homogeneous, \ie have similar compute service times $T_i$. In this case both algorithms are restricted to the slowest $n-z$ workers (see Lemma~\ref{lemma:z+1}), but PRAC incurs an $\epsilon$ overhead of tasks (due to using Fountain codes) which is not needed for Staircase codes. In particular, from \eqref{eq:T_{PRAC}} and \eqref{eq:sub}, when the $n-z$ slowest workers are homogeneous, the task completion time of PRAC and Staircase codes are equal to $\frac{b+\epsilon}{n-z} \E[\beta_{t,n}]$ and $\frac{b}{n-z} \E[\beta_{t,n}]$, respectively.

\section{Performance Evaluation} \label{sec:exp}
\subsection{Simulations}
In this section, we present simulations run on MATLAB, and compare PRAC with the following baselines: (i) Staircase codes \cite{BPR17}, (ii) C3P \cite{KS18} (which is not secure as it is not designed to be secure), and (iii) Genie C3P (GC3P) that extends C3P by assuming a knowledge of the identity of the eavesdroppers and ignoring them. We note that GC3P serves as a lower bound on private coded computing schemes for heterogeneous systems\footnote{If the system is homogeneous Staircase codes outperform GC3P, because pre-allocating tasks to the workers avoids the overhead needed by Fountain codes.} for the following reason: for a given number of $z$ colluding workers the ideal coded computing scheme knows which workers are eavesdroppers and ignores them to use the remaining workers without need of randomness. If the identity of the colluding workers is unknown, coded computing schemes require randomness and become limited to the $(z+1)^\text{st}$ fastest worker (Lemma~\ref{lemma:z+1}). GC3P and other coded computing schemes have similar performance if the $z$ colluding workers are the fastest workers. If the $z$ colluding workers are the slowest, then GC3P outperforms any coded computing scheme. Note that our solution PRAC considers the scenario of unknown eavesdroppers. Comparing PRAC with G3CP shows how good PRAC is as compared to the best possible solution for heterogeneous systems. In terms of comparing PRAC to solutions designed for the homogeneous setting, we restrict our attention to Staircase codes which are a class of secret sharing schemes that enjoys a flexibility in the number of workers needed to decode the matrix-vector multiplication. Staircase codes are shown to outperform any coded computing scheme that requires a threshold on the number of stragglers \cite{BPR17}.

\begin{figure*}[t]
\centering
 \setlength\figureheight{0.35\textwidth}
  \setlength\figurewidth{0.4\textwidth}
\begin{minipage}[b]{0.45\textwidth}
\centering
\captionsetup{subtype,width=\textwidth}
\resizebox{0.8\textwidth}{!}{
\input{NEWR1000_NumWorkvariable_NumColl1o4_iter100}
}
\caption{Task completion time as a function of the number of workers with $z=n/4$.}
\label{fig:thm2}
\end{minipage}\hfill %
\begin{minipage}[b]{0.45\textwidth}
  \centering
\captionsetup{subtype,width=\textwidth}
\resizebox{0.8\textwidth}{!}{
\input{NEWR1000_NumWorkvariable_NumColl13_iter100}
}
\caption{Task completion time as a function of the number of workers with $z=13$.}
\label{fig:thm3}
\end{minipage}
\caption{Comparison between PRAC, Staircase codes and GC3P in scenario~1 for different values of the number workers and number of colluding workers. We fix the number of rows to $m=1000$. For each value of the $x$-axis we run $100$ experiments and average the results. We observe that the difference between the completion time of PRAC and this of GC3P is large for small values of $n-z$ and decreases with the increase of $n-z$.}
\label{fig:mix}
\end{figure*}

\begin{figure*}[h]
\centering
 \setlength\figureheight{0.35\textwidth}
  \setlength\figurewidth{0.4\textwidth}
\begin{minipage}[b]{0.45\textwidth}
  \centering
\captionsetup{subtype,width=\textwidth}
\resizebox{0.8\textwidth}{!}{
\input{NEWR1000_NumWor50_NumCollvariable_iter100}
}
\caption{Task completion time as a function of the number of colluding workers for $n=50$. Computing time of the workers are chosen according to scenario~1.}
\label{fig:thm1}
\end{minipage}\hfill %
\begin{minipage}[b]{0.45\textwidth}
\centering
\captionsetup{subtype,width=\textwidth}
\resizebox{0.8\textwidth}{!}{
\input{Slowest_N_Z_Homo_NumCollvariable_NumWrok50_iter100.tex}
}
\caption{ Task completion time for $n=50$ workers and variable $z$. Computing times of the workers are chosen such that the $n-z$ slowest workers are homogeneous.}
\label{fig:thm2}
\end{minipage}\hfill %
\caption{Comparison between PRAC and Staircase codes average completion time as a function of number of colluding workers $z$. We fix the number of rows to $m=1000$. Both codes are affected by the increase of number of colluding helpers because their runtime is restricted to the slowest $n-z$ workers. We observe that PRAC outperforms Staircase codes except when the $n-z$ slowest workers are homogeneous.}
\label{fig:collusions}
\end{figure*}

In our simulations, we model the computation time of each worker $w_i$ by an independent shifted exponential random variable with rate $\lambda_i$ and shift $c_i$, \ie $F(T_i=t) = 1-\exp(-\lambda_i(t-c_i))$. We take $c_i=1/\lambda_i$ and consider three different scenarios for choosing the values of $\lambda_i$'s for the workers as follows:
\begin{itemize}
\item {\em Scenario~1}: we assign $\lambda_i = 3$ for half of the workers, then we assign $\lambda_i = 1$ for one quarter of the workers and assign $\lambda_i=9$ for the remaining workers.

\item {\em Scenario~2}: we assign $\lambda_i = 1$ for one third of the workers, the second third have $\lambda_i = 3$ and the remaining workers have $\lambda_i = 9$.

\item {\em Scenario~3}: we draw the $\lambda_i$'s independently and uniformly at random from the interval $[0.5,9]$. 
\end{itemize}

When running Staircase codes, we choose the parameter $k$ that minimizes the task completion time for the desired $n$ and $z$. We do so by simulating Staircase codes for all possible values of $k$, $z\leq k\leq n$, and choose the one with the minimum completion time.

We take $b=m$, \ie each row block is simply a row of $A$. The size of each element of $A$ and vector $\mathbf{x}$ are assumed to be 1 Byte (or 8 bits). Therefore, the size of each transmitted packet $p_{t,i}$ is $8*\ell$ bits. For the simulation results, we assume that the matrix $A$ is a square matrix, \ie $l=m$. We take $m=1000$, unless explicitly stated otherwise. $C_i$ denotes the average channel capacity of each worker $w_i$ and is selected uniformly from the interval $[10, 20]$ Mbps. The rate of sending a packet to worker $w_i$ is sampled from a Poisson distribution with mean $C_i$.  

In Figure~\ref{fig:remove} we show the effect of the number of rows $m$ on the completion time at the master. We fix the number of workers to $50$ and the number of colluding workers to $13$ and plot the completion time for PRAC, C3P, GC3P and Staircase codes. Notice that PRAC and Staircase codes have close completion time in scenario~1 (Figure~\ref{fig:thm1}) and this completion time is far from that of C3P. The reason is that in this scenario we pick exactly $13$ workers to be fast ($\lambda_i=9$) and the others to be significantly slower. Since PRAC assigns keys to the fastest $z$ workers, the completion time is dictated by the slow workers. To compare PRAC with Staircase codes notice that the majority of the remaining workers have $\lambda_i=3$ therefore pre-allocating equal tasks to the workers is close to adaptively allocating the tasks.

In terms of lower bound on PRAC, observe that when the fastest workers are assumed to be adversarial,  GC3P and PRAC have very similar task completion time. However, when the slowest workers are assumed to be adversarial the completion of GC3P is very close to C3P and far from PRAC. This observation is in accordance with Lemma~2. In scenarios~2 and 3 we pick the adversarial workers uniformly at random and observe that the completion time of PRAC becomes closer to GC3P when the workers are more heterogeneous. For instance, in scenario~3, GC3P and PRAC have closer performance when the workers' computing times are chosen uniformly at random from the interval $[0.5,9]$.


In Figure~\ref{fig:mix}, we plot the task completion time as a function of the number of workers $n$ for a fixed number of rows $m=1000$ and $\lambda_i$'s assigned according to scenario~1. In Figure~\ref{fig:mix}(a), we change the number of workers from $10$ to $100$ and keep the ratio $z/n = 1/4$ fixed. We notice that with the increase of $n$ the completion time of PRAC becomes closer to GC3P. In Figure~\ref{fig:mix}(b), we change the number of workers from $20$ to $100$ and keep $z = 13$ fixed. We notice that with the increase of $n$, the effect of the eavesdropper is amortized and the completion time of PRAC becomes closer to C3P. In this setting, PRAC always outperforms Staircase codes.

In Figure~\ref{fig:collusions}, we plot the task completion time as a function of the number of colluding workers. In Figure~\ref{fig:collusions}(a), we choose the computing time at the workers according to scenario~1. We change $z$ from $1$ to $40$ and observe that the completion time of PRAC deviates from that of GC3P with the increase of $z$. More importantly, we observe two inflection points of the average completion time of PRAC at $z=13$ and $z=37$. Those inflection points are due to the fact that we have $12$ fast workers ($\lambda = 9$) and $25$ workers with medium speed ($\lambda = 3$) in the system. For $z>36$, the completion time of Staircase codes becomes less than that of PRAC because the $14$ slowest workers are homogeneous. Therefore, pre-allocating the tasks is better than using Fountain codes and paying for the overhead of computations. To confirm that Staircase codes always outperforms PRAC when the slowest $n-z$ workers are homogeneous, we run a simulation in which we divide the workers into three clusters. The first cluster consists of $\lfloor z/2\rfloor$ fast workers ($\lambda = 9$), the second consists of $\lfloor z/2 \rfloor +1$ workers that are regular ($\lambda = 3$) and the remaining $n-z$ workers are slow ($\lambda = 1$). In Figure~\ref{fig:collusions}(b) we fix $n$ to $50$ and change $z$ from $1$ to $40$. We observe that Staircase codes always outperform PRAC in this setting. In contrast to non secure C3P, Staircase codes and PRAC are always restricted to the slowest $n-z$ workers and cannot leverage the increase of the number of fast workers. For GC3P, we assume that the fastest workers are eavesdroppers. We note that as expected from Lemma~\ref{lemma:z+1}, when the fastest workers are assumed to be eavesdroppers the performance of GC3P and PRAC becomes very close. 

\subsection{Experiments} 

{\em Setup.} The master device is a Nexus 5 Android-based smartphone running 6.0.1. The worker devices are Nexus 6Ps running Android 8.1.0. The master device connects to worker devices via Wi-Fi Direct links and the master is the group owner of Wi-Fi Direct group. The master device is required to complete one matrix multiplication ($y=Ax$) where $A$ is of dimensions $60 \times 10000$ and $x$ is a $10000\times 1$ vector. We also take $m=b$ \ie each packet is a row of $A$. We introduced an artificial delay at the workers following an exponential distribution. The introduced delays serves to emulate applications running in the background of the devices. A worker device sends the result to the master after it is done calculating and the introduced delay has passed. Furthermore, we assume that $z=1$ \ie there is one unknown worker that is adversarial among all the workers. The experiments are conducted in a lab environment where there are other Wi-Fi networks operating in the background.

{\em Baselines.} Our PRAC algorithm is compared to three baseline algorithms: (i)
Staircase codes that preallocate the tasks based on $n$, the number of workers, $k$, the minimum number of workers required to reconstruct the information, and $z$, the number of colluding workers;
(ii) GC3P in which we assume the adversarial worker is known and excluded during the task allocation;
(iii) Non secure C3P in which the security problem is ignored and the master device will utilize every resource without randomness. In this setup we run C3P on $n-z$ workers.

{\em Results.} 
Figure~\ref{exp_homogeneous} presents the task completion time with increasing number of workers for the homogeneous setup, \ie when all the workers have similar computing times.  
Computing delay for each packet 
follows an exponential distribution with mean $\mu = 1/\lambda = 3$ seconds in all workers. C3P performs the best in terms of completion time, but C3P does not provide any privacy guarantees.  PRAC outperforms Staircase codes when the number of workers is 5. 
The reason is that PRAC performs better than Staircase codes in heterogeneous setup, and when the number of workers increases, the system becomes a bit more heterogeneous.   
%
GC3P significantly outperforms PRAC in terms of completion time. Yet, it requires a prior knowledge of which worker is adversarial, which is often not available in real world scenarios.

\begin{figure}
\centering
\includegraphics[width=0.4\textwidth]{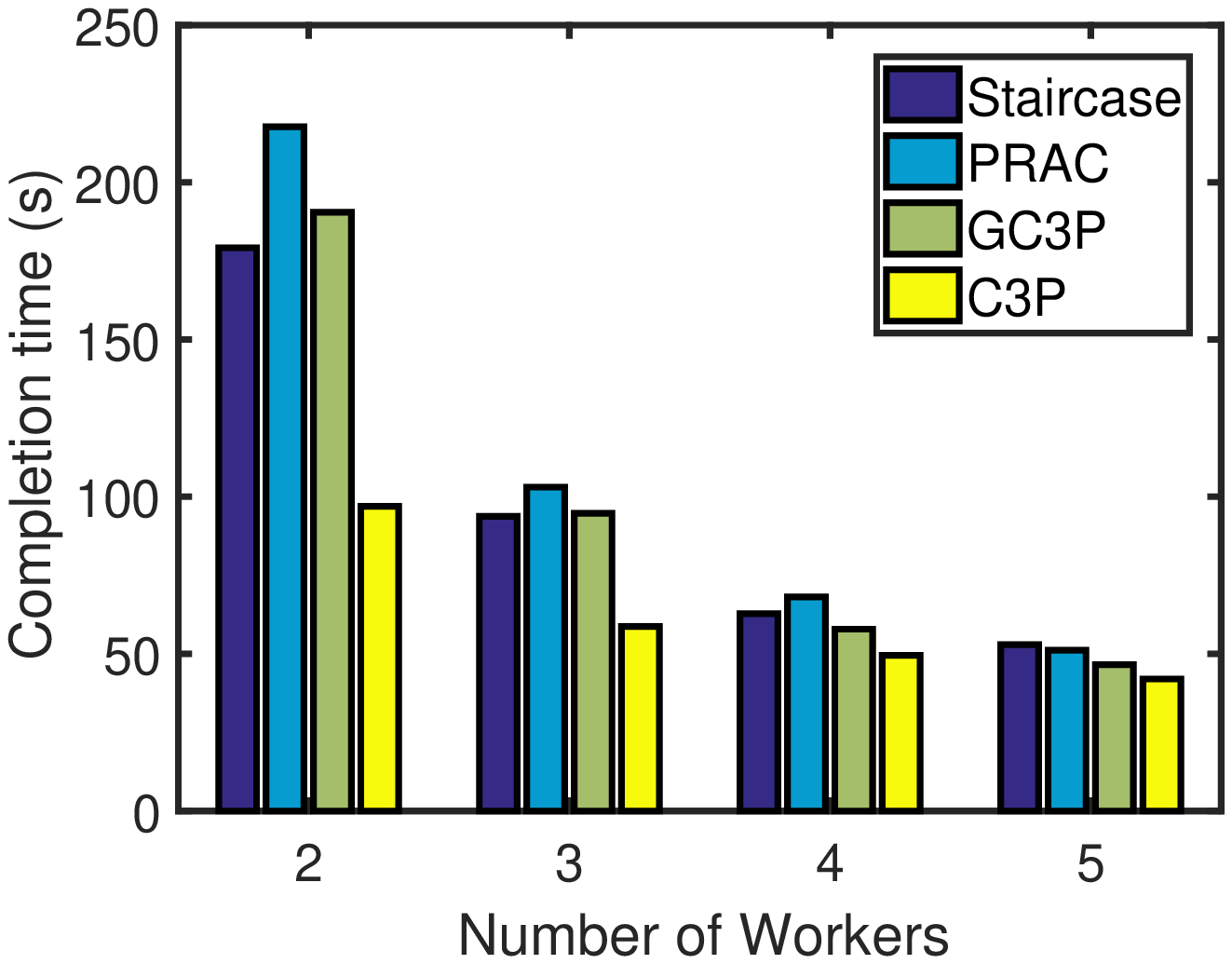}
\caption{Completion time as function of the number of workers in homogeneous setup.}
\label{exp_homogeneous}
\end{figure}

\begin{figure}
	\centering
	\begin{subfigure}{0.4\textwidth}
		\includegraphics[width=\textwidth]{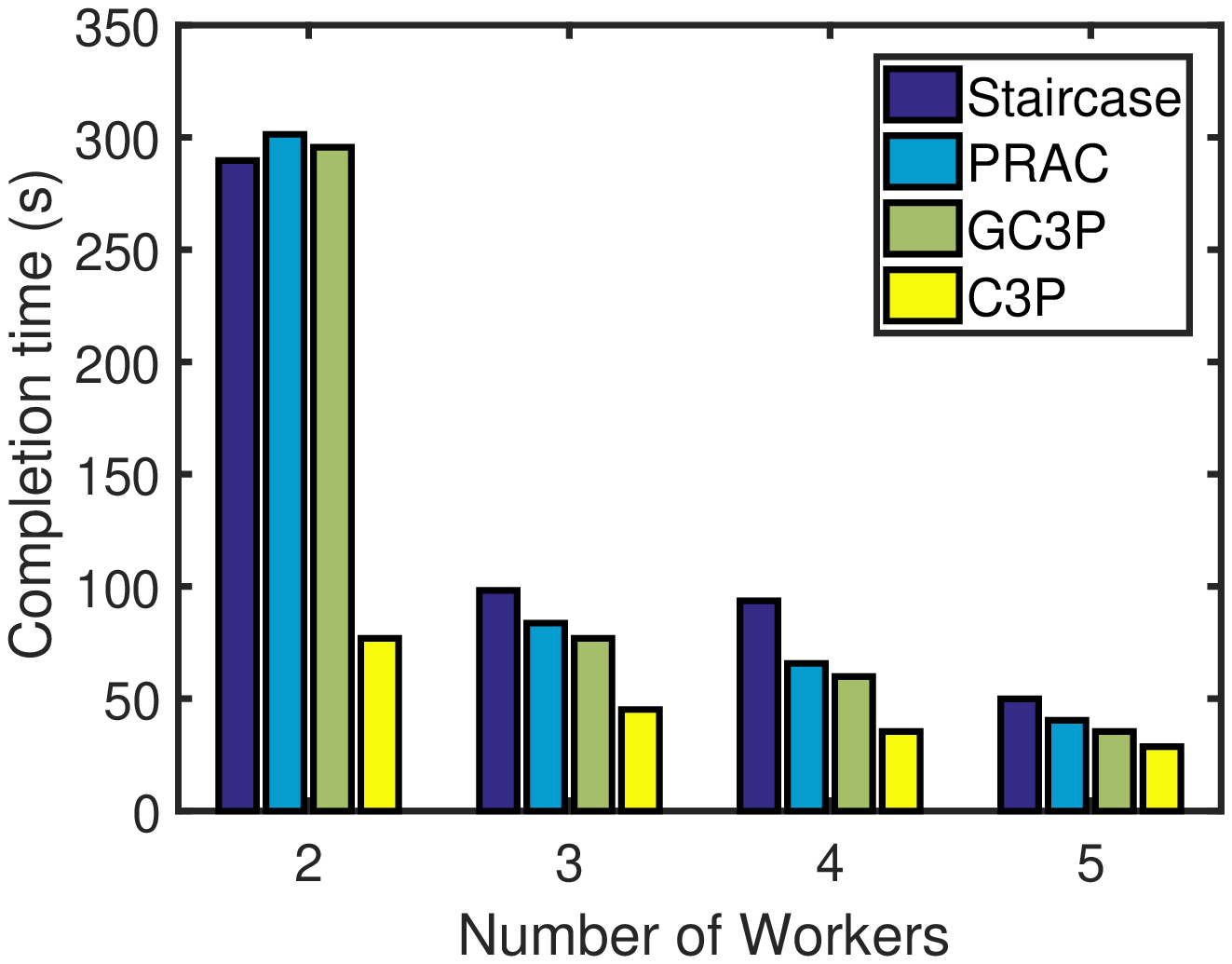}
		\caption{We assume a fast worker is adversarial for GC3P.}
        \label{exp_heterogeneous1}
	\end{subfigure}
	\vspace{1em}
	\begin{subfigure}{0.4\textwidth}
		\includegraphics[width=\textwidth]{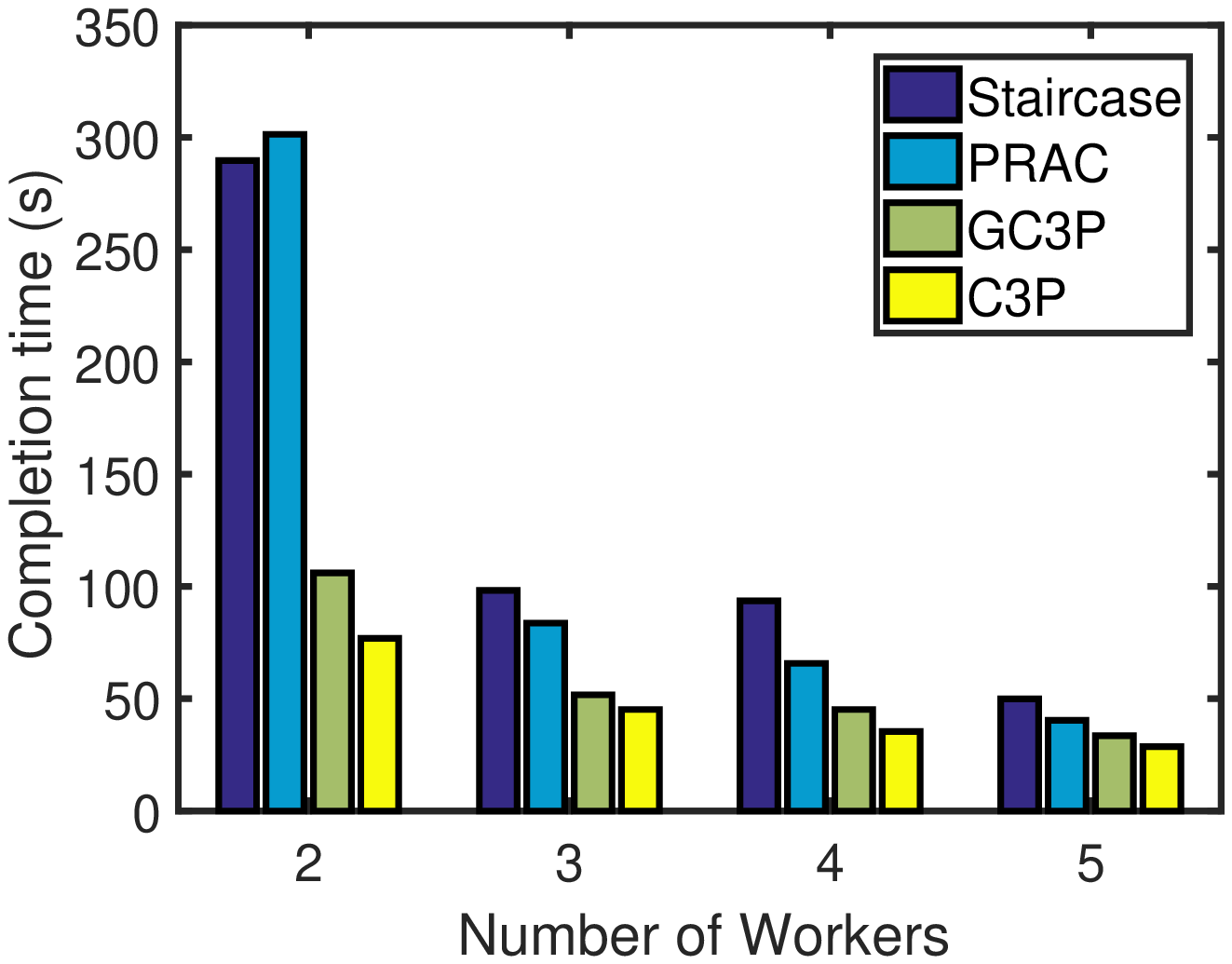}
		\caption{We assume a slow worker is adversarial for GC3P.}
        \label{exp_heterogeneous2}
	\end{subfigure}
	\caption{Completion time as function of the number of workers in heterogeneous setup.}
    \label{exp_heterogeneous_1}
\end{figure}

Now, we focus on heterogeneous setup. We group the workers into two groups; fast workers (per task delay follows exponential delay with mean $2$ seconds) and slow workers (per task delay follows exponential distribution with mean $5$ seconds). Figure \ref{exp_heterogeneous_1} presents the completion time as a function of number of workers. In this setup, for the $n$-worker scenario, there are $\left \lceil{\frac{n}{2}}\right \rceil$ fast and $\left \lfloor{\frac{n}{2}}\right \rfloor$ slow workers. 
%
The difference between the setups of Figure~\ref{exp_heterogeneous_1}(a) and Figure~\ref{exp_heterogeneous_1}(b) is that we remove a fast worker (as adversarial) for GC3P in the former, whereas in the latter, we assume that the eavesdropper is a slow worker. As illustrated in Figure~\ref{exp_heterogeneous_1}, for the 2-worker case,
%
%
due to the $5\%$  overhead introduced by Fountain codes, PRAC performs worse than Staircase code. However, PRAC outperforms Staircase codes in terms of completion time for 3, 4, and 5 worker cases. This is due to the fact that PRAC can utilize results calculated by slow workers more effectively when the number of workers is large. On the other hand, the results computed by slow workers are often discarded in Staircase codes, which is a waste of computation resources. If a fast worker is removed as adversarial for GC3P, the difference between the performance of GC3P and PRAC becomes smaller. This result is intuitive as, in PRAC, the master has to wait for the $(z+1)^\text{st}$ fastest worker to decode $\mba \mbx$, which is also the case for GC3P in this setting. 

\begin{figure}
	\centering
	\begin{subfigure}{0.4\textwidth} 
		\includegraphics[width=\textwidth]{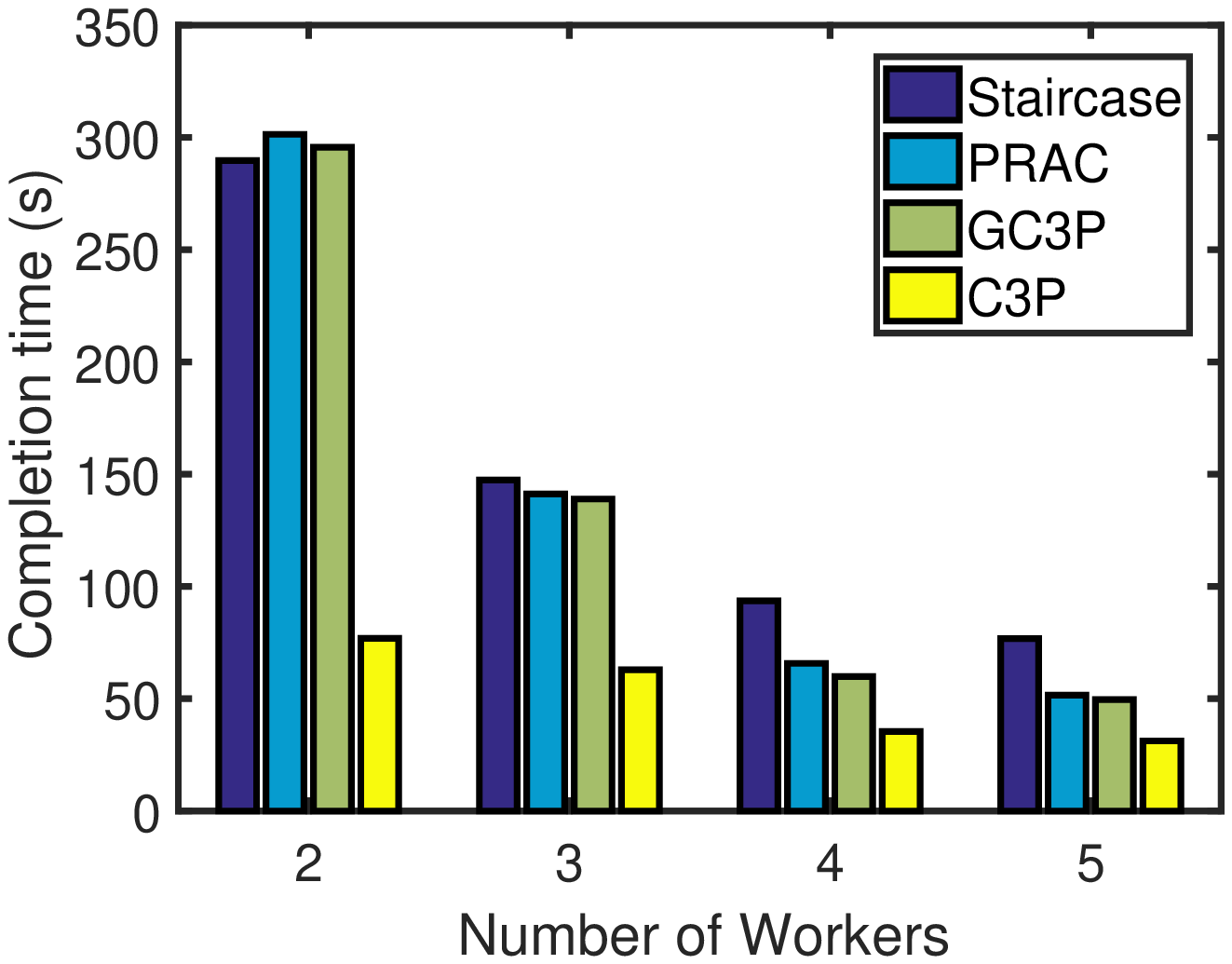}
		\caption{We assume a fast worker is adversarial for GC3P.} 
        \label{exp_heterogeneous3}
	\end{subfigure}
	\vspace{1em} 
	\begin{subfigure}{0.4\textwidth} 
		\includegraphics[width=\textwidth]{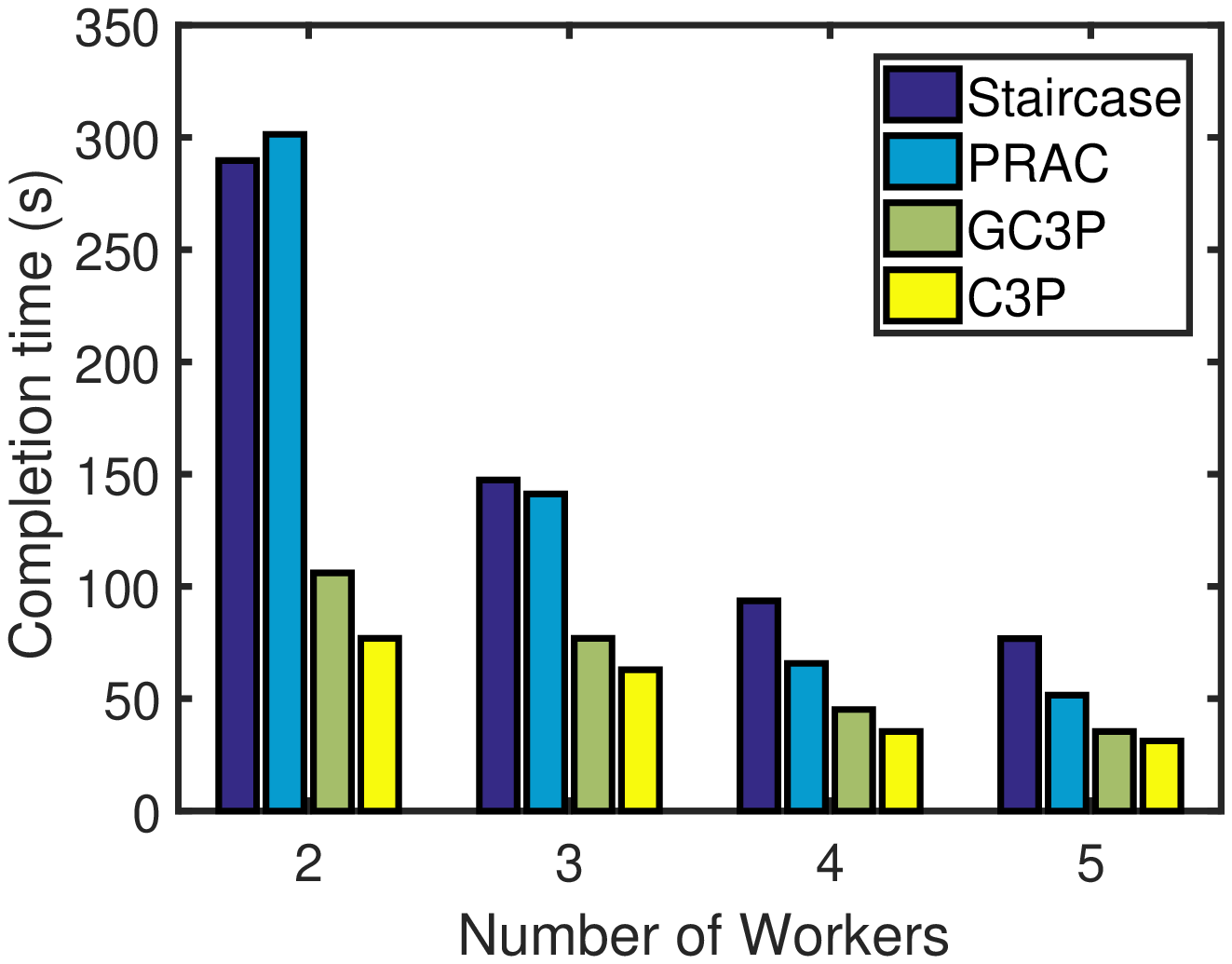}
		\caption{We assume a slow worker is adversarial for GC3P.} 
        \label{exp_heterogeneous4}
	\end{subfigure}
	\caption{Completion time as function of the number of workers in heterogeneous setup.}
    \label{exp_heterogeneous_2}
\end{figure}

In Figure~\ref{exp_heterogeneous_2}, we consider the same setup with the exception that 
for the $n$-worker scenario, there are $\left \lceil{\frac{n}{2}}\right \rceil$ slow and $\left \lfloor{\frac{n}{2}}\right \rfloor$ fast workers.
%
%
Staircase codes perform more closely to PRAC in the 3-worker case as compared to Figure~\ref{exp_heterogeneous_1} since the setup of Fig.6 assumes that the n-z=2 slowest workers are homogeneous, whereas in Fig.5 the n-z=2 slowest workers are heterogeneous. Yet, for 5-worker case, PRAC outperforms Staircase codes when comparing to Figure~\ref{exp_heterogeneous_1} since PRAC is adaptive to time-varying resources while Staircase codes assigns tasks a priori in a static manner. 

Note that in all experiments when $n-z$ slowest workers are homogeneous Staircase codes outperform GC3P and PRAC. This happens because pre-allocating the tasks to the workers avoids the overhead of sub-tasks required by Fountain codes and utilizes all the workers to their fullest capacity.

\section{Related work} \label{sec:related}
Mobile cloud computing is a rapidly growing field with the aim of providing better experience of quality and extensive computing resources to mobile devices \cite{dinh2013survey, FERNANDO201384}. The main solution to mobile computing is to offload tasks to the cloud or to neighboring devices by exploiting connectivity of the devices. With task offloading come several challenges such as heterogeneity of the devices, time varying communication channels and energy efficiency, see \eg \cite{6517049, 7437124, 6903299, 7037232}. We refer interested reader to \cite{KS18} and references within for a detailed literature on edge computing and mobile cloud computing. 

The problem of stragglers in distributed systems is initially studied by the distributed computing community, see \eg 
\cite{dean2008mapreduce,DB13,recht2011hogwild,dean2012large}. Research interest in using coding theoretical techniques for straggler mitigation in distributed content download and distributed computing is rapidly growing. The early body of work focused on content download, see \eg \cite{HPZR12,JLS12,TOFEC,KSS15, PS18}. Using codes for straggler mitigation in distributed computing started in~\cite{lee2018speeding} where the authors proposed the use of MDS codes for distributed linear machine learning algorithms in homogeneous workers setting. 

Following the work of \cite{lee2018speeding}, coding schemes for straggler mitigation in distributed matrix-matrix multiplication, coded computing and machine learning algorithms are introduced and the fundamental limits between the computation load and the communication cost are studied, see \eg \cite{yu2017polynomial,QMA18} and references within for matrix-matrix multiplication, see \cite{karakus2017straggler, maity2018robust, chen2018draco, kiani2018exploitation, ozfaturay2018speeding, ye2018communication, ferdinand2018anytime, li2018near, lee2018speeding, tandon2017gradient,raviv2017gradient,halbawi2017improving,li2016unified,li2016fundamental} for machine learning algorithms and \cite{DCG16, yang2017computing,dutta2017coded, USMHYJTP18} and references within for other topics. 

Codes for privacy and straggler mitigation in distributed computing are first introduced in \cite{BPR17} where the authors consider a homogeneous setting and focus on matrix-vector multiplication. Beyond matrix-vector multiplication, the problem of private distributed matrix-matrix multiplication and private polynomial computation with straggler tolerance is studied \cite{d2018gasp,yu2018lagrange,chang2018capacity,kakar2018rate,yang2018secure,YRSA18}. The former works are designed for the homogeneous static setting in which the master has a prior knowledge on the computation capacities of the workers and pre-assigns the sub-tasks equally to them. In addition, the master sets a threshold on the number of stragglers that it can tolerate throughout the whole process. In contrast, PRAC is designed for the heterogeneous dynamic setting in which workers have different computation capacities that can change over time. PRAC assigns the sub-tasks to the workers in an adaptive manner based on the estimated computation capacity of each worker. Furthermore, PRAC can tolerate a varying number of stragglers as it uses an underlying rateless code, which gives the master a higher flexibility in adaptively assigning the sub-tasks to the workers. Those properties of PRAC allow a better use of the workers over the whole process. On the other hand, PRAC is restricted to matrix-vector multiplication. Although coded computation is designed for linear operations, there is a recent effort to apply coded computation for non-linear operations. For example, \cite{so2019codedprivateml} applied coded computation to logistic regression, and the framework of Gradient coding started in \cite{tandon2017gradient} generalizes to any gradient-descent algorithm. Our work is complementary with these works. For example, our work can be directly used as complementary to \cite{so2019codedprivateml} to provide privacy and adaptive task offloading to logistic regression.

Secure multi-party communication (SMPC) \cite{SMPCbook} can be related to our work as follows. The setting of secure multi-party computing schemes assumes the presence of several parties (masters in our terminology) who want to compute a function of all the data owned by the different parties without revealing any information about the individual data of each party. This setting is a generalized version of the master/worker setting that we consider. More precisely, an SMPC scheme reduces to our Master/worker setting if we assume that only one party owns data and the others have no data to include in the function to be computed. SMPC schemes use threshold secret sharing schemes, therefore they restrict the master to a fixed number of stragglers. Thus, showing that PRAC outperforms Staircase codes (which are the best known family of threshold secret sharing schemes) implies that PRAC outperform the use of SMPC schemes that are reduced to this setting. Works on privacy-preserving machine learning algorithms are also related to our work. However, the privacy constraint in this line of work is computational privacy and the proposed solutions do not take stragglers into account, see \eg \cite{secureml1,secureml2,gade2017private}. 

We restrict the scope of this paper to eavesdropping attacks, which are important on their own merit. Privacy and security can be achieved by using Maximum Distance Separable (MDS)-like codes which restrict the master to a fixed maximum number of stragglers \cite{yang2018secure,yu2018lagrange}. Our solution on the other hand addresses the privacy problem in an adaptive coded computation setup without such a restriction. In this setup, security cannot be addressed by expanding the results of \cite{yang2018secure,yu2018lagrange}. In fact, we developed a secure adaptive coded computation mechanism in our recent paper \cite{keshtkarjahromi2019secure} against Byzantine attacks. The private and secure adaptive coded computation obtained by combining this paper and \cite{keshtkarjahromi2019secure} is out of scope of this paper.

\section{Conclusion}\label{sec:conc}

The focus of this paper is to develop a secure edge computing mechanism to mitigate the computational bottleneck of IoT devices by allowing these devices to help each other in their computations, with possible help from the cloud if available. Our key tool is the theory of coded computation, which advocates mixing data in computationally intensive tasks by employing erasure codes and offloading these tasks to other devices for computation. Focusing on eavesdropping attacks, we designed a private and rateless adaptive coded computation (PRAC) mechanism considering (i) the privacy requirements of IoT applications and devices, and (ii) the heterogeneous and time-varying resources of edge devices. Our proposed PRAC model can provide adequate security and latency guarantees to support real-time computation at the edge. We showed through analysis, MATLAB simulations, and experiments on Android-based smartphones that PRAC outperforms known secure coded computing methods when resources are heterogeneous.

\section{Acknowledgement}
This work was supported in parts by the Army Research Lab (ARL) under Grant W911NF-1820181, National Science Foundation (NSF) under Grants CNS-1801708 and CNS-1801630, and the National Institute of Standards and Technology (NIST) under Grant 70NANB17H188.

\bibliographystyle{ieeetr}
\bibliography{SCCC}

\appendices
\section{Hiding the Vector $\mathbf{x}$}
\label{appen:hiding}
In machine learning applications, the master runs iterative algorithms in which the vector $\mathbf{x}$ contains information about $A$ and needs to be hidden from the workers. We describe how PRAC can be generalized to achieve privacy for both $A$ and $\mathbf{x}$. The idea is to divide the $n$ workers into two disjoint groups and ask each of them to privately multiply $A$ by a vector that is statistically independent of $\mathbf{x}$. In addition, the master should be able to decode $A\mathbf{x}$ from the results of both multiplications. The scheme works as follows. The master divides the workers into two groups of cardinality $n_1$ and $n_2$ such that $n_1+n_2=n$ and chooses the security parameters $z_1<n_1$ and $z_2<n_2$. To hide $\mathbf{x}$, the master generates a random vector $\mathbf{u}$ of same size as $\mathbf{x}$ and sends $\mathbf{x}+\mathbf{u}$ to the first group and $\mathbf{u}$ to the second group. Afterwards, the master applies PRAC on both groups. According to our scheme, the master decodes $A(\mathbf{x}+\mathbf{u})$ and $A\mathbf{u}$ after receiving enough responses from the workers of each group. Hence, the master can decode $A\mathbf{x}$. Note that no information about $\mathbf{x}$ is revealed because it is one-time padded by $\mathbf{u}$. Note that here we assume workers from group~1 do not collude with workers from group~2. The same idea can be generalized to the case where workers from different groups can collude by creating more groups and encoding $\mbx$ using an appropriate secret sharing scheme. For instance, if the master divides the workers into $3$ groups and workers from any $2$ different groups can collude, the master encodes $\mbx$ into $\mathbf{u}_1$, $\mathbf{u}_2$ and $\mathbf{u}_1+\mathbf{u}_2+\mbx$ and sends each vector to a different group. 

\section{\label{appen:proofPrivacy} Extension of proof of privacy (\ie Theorem~\ref{thm:pracPrivacy}) } 
Since at each round we generate new random matrices, it is enough to study the privacy condition at one round. Consider a given round $t$ of PRAC. Let $P_i$ denote the random variable representing packet $p_{i}$ sent to worker $w_i$. For any subset $Z\subset \{1,\dots,n\}, |Z|=z$, denote by $P_Z$ the collection of packets indexed by $Z$, \ie $P_Z=\{p_i; i \in Z\}$. We prove that the perfect secrecy constraint $H(A \mid P_Z)=H(A)$, given in~\eqref{eq:priv}, is equivalent to $H(K \mid P_Z,A)=0$. The proof is standard \cite{RSKGlobecom11,RSS12,BRIT17} but we reproduce it here for completeness.  In what follows, the logarithms in the entropy function are taken base $q$, where $q$ is a power of prime for which all matrices can be defined in a finite field $\mathbb{F}_q$. We can write,
\begin{align}
H(A \mid P_Z)&=H(A)-H(P_Z)+H(P_Z \mid A)\\
~&=H(A)-H(P_Z)+H(P_Z \mid A)-H(P_Z \mid A, K) \label{eq:sk}\\
~&=H(A)-H(P_Z)+I(P_Z; K \mid A)\\
~&=H(A)-H(P_Z)+H(K \mid A) - H(K \mid P_Z,A)\\
~&=H(A)-H(P_Z)+H(K) - H(K \mid P_Z,A) \label{eq:keys}\\
~&=H(A)-z+z - H(K \mid P_Z,A) \label{eq:sc}\\
~&=H(A) - H(K \mid P_Z,A).
\end{align}
\vspace{0.2cm}

\noindent Equation~\eqref{eq:sk} follows from the fact that given the data $A$ and the keys $R_{1}, \dots, R_z$ all packets generated by the master can be decoded, in particular the packets $P_Z$ received by any $z$ workers can be decoded, 
\ie  $H(P_Z \mid A, K)=0$. Equation~\eqref{eq:keys} follows because the random matrices are chosen independently from the data matrix $A$ and equation~\eqref{eq:sc} follows because PRAC uses $z$ independent random matrices that are chosen uniformly at random from the field $\mathbb{F}_q$. Therefore, proving that $H(A|P_Z) = H(A)$ is equivalent to proving that $H(K \mid P_Z,A) = 0$. In other words, we need to prove that the random matrices can be decoded given the collection of packets sent to any $z$ workers and the data matrix $A$. This is the main reason behind encoding the random matrices using an $(n,z)$ MDS code. We formally prove that $H(K \mid P_Z,A) = 0$ in the proof of Theorem~\ref{thm:pracPrivacy}. Note from equation~\eqref{eq:keys} that for any code to be information theoretically private, $H(K)$ cannot be less then $H(P_Z)=z$. This means that a secure code must use at least $z$ independent random matrices.

\section{\label{appen:proofLm} Proof of Lemma~\ref{lemma:z+1}}

We prove the lemma by contradiction. Assume that there exists a private coded computing scheme for distributed linear computation that is secure against $z$ colluding workers and allows the master to decode $A\mathbf{x}$ using the help of the fastest $z$ workers. Without loss of generality, assume that the workers are ordered from the fastest to the slowest, \ie worker $w_1$ is the fastest and worker $w_n$ is the slowest. The previous assumption implies that the results sent from the first $z$ workers contain information about $\mba \mbx$, otherwise the master would have to wait at least for the $(z+1)^\text{st}$ fastest worker to decode $\mba \mbx$.  
By linearity of the multiplication $A\mathbf{x}$, decoding information about $A\mathbf{x}$ from the results of $z$ workers implies decoding information about $A$ from the packets sent to those $z$ workers. Hence, there exists a set of $z$ workers for which $H(S|{P}_Z) \neq 0$, where ${P}_Z$ denotes the tasks allocated to a subset $Z\subset \{1,\dots, n\}$ of $z$ workers, hence violating the privacy constraint. Therefore, any private coded computing scheme for linear computation limits the master to the speed of the $(z+1)^\text{st}$ fastest worker in order to decode the wanted result.

\section{\label{appen:proofTh3} Proof of Theorem~\ref{thm:pracCompTime}}


The total delay for receiving $\tau_i$ computed packets from worker $w_i$ is equal to $$T_i \approx RTT_i+\tau_i \E[\beta_{t,i}] \approx \tau_i \E[\beta_{t,i}]$$ where $RTT_i$ is the average transmission delay for sending one packet to worker $w_i$ and receiving one computed packet from the worker, $\beta_{t,i}$ is the computation time spent on multiplying packet $p_{t,i}$ by $\mathbf{x}$ at worker $w_i$, and the average $\E[\beta_{t,i}]$ is taken over all $\tau_i$ packets. The reason is that PRAC is a dynamic algorithm that sends packets to each worker $w_i$ with the interval of $\E[\beta_{t,i}]$ between each two consecutive packets and it utilizes the resources of workers fully \cite{KS18LongerVersion}. 
The reason behind counting only one round-trip time (RTT) in $T_i$ is that in PRAC, the packets are being transmitted to the workers while the previously transmitted packets are being computed at the worker. Therefore, in the overall delay only one $RTT_i$ is required for sending the first packet $p_{1,i}$ to worker $w_i$ and receiving the last computed packet $p_{\tau_i,i}\mathbf{x}$ at the master. To approximate the total delay, we assume that the transmission delay of one packet is negligible compared to the computing delay of all $\tau_i$ packets, which is a valid assumption in practice for IoT-devices at the edge.

On the other hand, in PRAC, the master stops sending packets to workers as soon as it collectively receives $b+\epsilon$ computed packets from the $n-z$ slowest workers (note that $b+\epsilon$ is the number of computed packets required for successful decoding, where $\epsilon$ is the overhead due to  Fountain Coding), \ie $\sum_{i=z+1}^{n}{\tau_i}=b+\epsilon$. Note that the $z$ fastest workers are assigned for computing the keys as described in the previous sections. Due to efficiently using the resources of workers by PRAC, all $n-z$ workers will finish computing $\tau_i$ packets approximately at the same time, \ie $T_{PRAC}\approx T_i \approx \tau_i \E[\beta_{t,i}], i=z+1,...,n$. By replacing $\tau_i$ with $\frac{T_{PRAC}}{\E[\beta_{t,i}]}$ in $\sum_{i=z+1}^{n}{\tau_i}=b+\epsilon$, we can show that $T_{PRAC} \approx \frac{b+\epsilon}{\sum_{i=z+1}^n 1/\E[\beta_{t,i}]}$. Note that the approximated value approaches the exact value by increasing $b$. The reason is that the workers' efficiency increases with increasing $b$
.

\section{\label{appen:proofTh4} Proof of Theorem~\ref{th:gap_PRAC_staircase}}


We express $\mathbb{E}[T_{\text{SC}}]$ as a function of the computing time $\beta_{t, i}$ of worker $w_i$, $i=1,\dots,n$, as
\begin{align}
\E[T_{\text{SC}}] = &\min_{d\in\{k,\dots,n\}}\left\{\dfrac{k-z}{d-z}\E[T_{(d)}]\right\} \\
= &\min_{d\in\{k,\dots,n\}}\left\{\dfrac{b}{d-z}\E[\beta_{t,d}]\right\},
\end{align}
where $w_d$ is the $d^\text{th}$ fastest worker. Next, we find a lower bound on $\E[T_{\text{SC}}]-\E\left[T_{PRAC}\right]$ as follows
\begin{align}
\E[T_{\text{SC}}]&-\E\left[T_{PRAC}\right] = \frac{b}{\frac{d-z}{\E[\beta_{t,d}]}}-\frac{b+\epsilon}{\sum_{i=z+1}^n \frac{1}{\E[\beta_{t,i}]}}\\
=& \frac{b}{\frac{d-z}{\E[\beta_{t,d}]}}-\frac{b+\epsilon}{\sum_{i=z+1}^d \frac{1}{\E[\beta_{t,i}]}+\sum_{i=d+1}^n \frac{1}{\E[\beta_{t,i}]}}\\
\geq & \frac{b}{\frac{d-z}{\E[\beta_{t,d}]}}-\frac{b+\epsilon}{(d-z) \frac{1}{\E[\beta_{t,d}]}+(n-d) \frac{1}{\E[\beta_{t,n}]}}\label{inequality}\\
=& \frac{\frac{b(n-d)}{\E[\beta_{t,n}]}-\frac{\epsilon(d-z)}{\E[\beta_{t,d}]}}{\frac{d-z}{\E[\beta_{t,d}]}(\frac{d-z}{\E[\beta_{t,d}]}+\frac{n-d}{\E[\beta_{t,n}]})}\\
=& \frac{bx-\epsilon y}{y(x+y)},
\end{align}
where $x=\frac{n-d}{\E[\beta_{t,n}]}$ and $y=\frac{d-z}{\E[\beta_{t,d}]}$ and the inequality (\ref{inequality}) comes from the fact that $z \leq k \leq d \leq n$ and the workers are ordered from the fastest to the slowest.

\end{document}

%% file: S2removeslowest_Rvariable_NumWork50_NumColl13_iter100.tex
%
%
\definecolor{mycolor1}{rgb}{0.0, 0.58, 0.71}
\definecolor{mycolor3}{rgb}{1.0, 0.65, 0.0}
\definecolor{mycolor4}{rgb}{0.53, 0.66, 0.42} 
\definecolor{mycolor2}{rgb}{0.8, 0.31, 0.36}
\begin{tikzpicture}
\def \msize{3}
\begin{axis}[%
width=0.951\figurewidth,
height=\figureheight,
at={(0\figurewidth,0\figureheight)},
scale only axis,
xmin=0,
xmax=1000,
xlabel={\Large Number of rows in A},
ymin=0,
ymax=32,
ylabel={\Large Average completion time},
axis background/.style={fill=white},
legend style={legend cell align=left,align=left,draw=white!15!black, at={(0.5,1)}}
]

\addplot [color=mycolor4,dashed,ultra thick,mark=+,mark options={solid}, mark size = \msize]
  table[row sep=crcr]{%
100	3.07398777527893\\
200	5.9766894033533\\
300	8.87452854633174\\
400	11.7473597273745\\
500	14.6789673632457\\
600	17.5520402477864\\
700	20.4667032220259\\
800	23.3956289228255\\
900	26.2825598706455\\
1000	29.256015236007\\
1100	32.1850152723158\\
1200	35.0577584657362\\
1300	38.0001151770413\\
1400	40.9451665828924\\
1500	43.8252297595605\\
1600	46.7227218538122\\
1700	49.6315796282632\\
1800	52.5240150756498\\
1900	55.4877006507744\\
2000	58.4077288294747\\
};
\addlegendentry{Staircase codes};

\addplot [color=mycolor3,dashed,ultra thick,mark=o,mark options={solid}, mark size = \msize]
  table[row sep=crcr]{%
100	2.97957420691925\\
200	5.37479685335549\\
300	7.902145887417\\
400	10.3515010577231\\
500	12.7675390786796\\
600	15.2078841199644\\
700	17.802173122435\\
800	20.2338982598432\\
900	22.735865825799\\
1000	25.243437566318\\
1100	27.6988095621937\\
1200	30.1823955477395\\
1300	32.5290395907913\\
1400	34.9840748569972\\
1500	37.6931083858944\\
1600	40.0034949848464\\
1700	42.5121370951853\\
1800	45.0400531680891\\
1900	47.5341001805296\\
2000	50.0730606250933\\
};
\addlegendentry{PRAC};

\addplot [color=mycolor2,solid,mark=+,mark options={solid}, ultra thick,mark size = \msize]
  table[row sep=crcr]{%
100	2.72729081205054\\
200	5.10059297161836\\
300	7.58823878687483\\
400	9.97604087778761\\
500	12.3237219644305\\
600	14.7107195448541\\
700	17.2529799372984\\
800	19.6054624913945\\
900	22.0461642439862\\
1000	24.498280411757\\
1100	26.9005303451615\\
1200	29.3047331916657\\
1300	31.6047838544385\\
1400	34.0067415088976\\
1500	36.6198970964271\\
1600	38.8977057865226\\
1700	41.3426709976467\\
1800	43.8306164730546\\
1900	46.2373218115397\\
2000	48.6988160503474\\
};
\addlegendentry{GC3P 1};

\addplot [color=mycolor2,solid,ultra thick,mark=o,mark options={solid},mark size = \msize]
  table[row sep=crcr]{%
100	1.29060453046032\\
200	2.45233192947778\\
300	3.59455182297876\\
400	4.73984506170949\\
500	5.87822656823061\\
600	7.02194852038843\\
700	8.20082943646749\\
800	9.33568720987128\\
900	10.4957251509212\\
1000	11.6410825789584\\
1100	12.7957223874547\\
1200	13.9633535939479\\
1300	15.0671299795819\\
1400	16.2445700081255\\
1500	17.3627840325308\\
1600	18.5252929463342\\
1700	19.6833023757855\\
1800	20.8422354256425\\
1900	22.0191367237434\\
2000	23.1446059070414\\
};
\addlegendentry{GC3P 2};

\addplot [color=mycolor1,dashed, ultra thick,mark=diamond,mark options={solid} , mark size = \msize]
  table[row sep=crcr]{%
100	1.26147444457214\\
200	2.32519853745949\\
300	3.42018673275339\\
400	4.4724163782977\\
500	5.51932365851534\\
600	6.61445851741387\\
700	7.68637679184898\\
800	8.76433213376233\\
900	9.84002707091013\\
1000	10.9249597607753\\
1100	11.9778081736427\\
1200	13.0594316013864\\
1300	14.0969062870362\\
1400	15.1840859059524\\
1500	16.3288951904394\\
1600	17.3105835426314\\
1700	18.3938110855208\\
1800	19.4955743730788\\
1900	20.6279582678162\\
2000	21.6708357445922\\
};
\addlegendentry{C3P};

\end{axis}
\end{tikzpicture}%

%% file: S2removefastest_Rvariable_NumWork50_NumColl13_iter100.tex
%
%
\definecolor{mycolor1}{rgb}{0.0, 0.58, 0.71}
\definecolor{mycolor3}{rgb}{1.0, 0.65, 0.0}
\definecolor{mycolor4}{rgb}{0.53, 0.66, 0.42} 
\definecolor{mycolor2}{rgb}{0.8, 0.31, 0.36}
\begin{tikzpicture}
\def \msize{3}
\begin{axis}[%
width=0.951\figurewidth,
height=\figureheight,
at={(0\figurewidth,0\figureheight)},
scale only axis,
xmin=0,
xmax=1000,
xlabel={\Large Number of rows in A},
ymin=0,
ymax=35,
ylabel={\Large Average completion time},
axis background/.style={fill=white},
legend style={legend cell align=left,align=left,draw=white!15!black, at={(0.5,1)}}
]

\addplot [color=mycolor4,dashed,ultra thick,mark=+,mark options={solid}, mark size = \msize]
  table[row sep=crcr]{%
100	3.65062417283511\\
200	7.12908941481\\
300	10.5984003911704\\
400	13.9904653048975\\
500	17.4628866447763\\
600	20.9551095805497\\
700	24.3488011602751\\
800	27.8872414482911\\
900	31.3727063533255\\
1000	34.8828948372145\\
};
\addlegendentry{Staircase codes};

\addplot [color=mycolor3,dashed,ultra thick,mark=o,mark options={solid}, mark size = \msize]
  table[row sep=crcr]{%
100	3.05689312208419\\
200	4.99049084018271\\
300	7.23081593612045\\
400	9.42374041894448\\
500	11.6424570330478\\
600	13.9679346953894\\
700	16.2021284596501\\
800	18.3829424271816\\
900	20.610755243029\\
1000	22.8889556363238\\
};
\addlegendentry{PRAC};

\addplot [color=mycolor2,solid,mark=+,mark options={solid}, ultra thick,mark size = \msize]
  table[row sep=crcr]{%
100	1.51646226073011\\
200	2.85451522631444\\
300	4.22307554584364\\
400	5.53398716041868\\
500	6.87065440674264\\
600	8.31514878965861\\
700	9.66122209532708\\
800	10.9997090515533\\
900	12.1569628261343\\
1000	13.7938347814301\\
};
\addlegendentry{GC3P};

\addplot [color=mycolor1,dashed, ultra thick,mark=diamond,mark options={solid} , mark size = \msize]
  table[row sep=crcr]{%
100	1.16994657804265\\
200	2.14517683973822\\
300	3.15419863321281\\
400	4.12947494468375\\
500	5.12959410275999\\
600	6.11487067594198\\
700	7.1276803727168\\
800	8.11308761896824\\
900	9.07783270876941\\
1000	10.1000448504055\\
};
\addlegendentry{C3P};

\end{axis}
\end{tikzpicture}%

%% file: S4nccpRvariable_NumWork50_NumColl13_iter100.tex
%
%
\definecolor{mycolor1}{rgb}{0.0, 0.58, 0.71}
\definecolor{mycolor3}{rgb}{1.0, 0.65, 0.0}
\definecolor{mycolor4}{rgb}{0.53, 0.66, 0.42} 
\definecolor{mycolor2}{rgb}{0.8, 0.31, 0.36}
\begin{tikzpicture}
\def\msize{3}
\begin{axis}[%
width=0.951\figurewidth,
height=\figureheight,
at={(0\figurewidth,0\figureheight)},
scale only axis,
xmin=0,
xmax=1000,
xlabel={\Large Number of rows in A},
ymin=0,
ymax=27,
ylabel={\Large Average completion time},
axis background/.style={fill=white},
legend style={legend cell align=left,align=left,draw=white!15!black, at={(0.5,1)}}
]

\addplot [color=mycolor4,dashed,ultra thick,mark=+,mark options={solid}, mark size = \msize]
  table[row sep=crcr]{%
100	2.7542725108681\\
200	5.42102600399259\\
300	8.2162363191682\\
400	10.6654119869774\\
500	13.9068019206387\\
600	16.4945035291028\\
700	19.3554064522527\\
800	22.1313778506127\\
900	25.2608239080971\\
1000	26.6395360741626\\
};
\addlegendentry{Staircase codes};

\addplot [color=mycolor3,dashed,ultra thick,mark=o,mark options={solid}, mark size = \msize]
  table[row sep=crcr]{%
100	2.20568371519631\\
200	3.51543174496357\\
300	5.01728761048299\\
400	6.50699777863379\\
500	8.22561901210747\\
600	9.75882471440259\\
700	11.3424244083203\\
800	12.990014880293\\
900	14.7226400212119\\
1000	15.6649018109494\\
};
\addlegendentry{PRAC};
\addplot [color=mycolor2,solid,mark=+,mark options={solid}, ultra thick,mark size = \msize]
  table[row sep=crcr]{%
100	1.37061891017534\\
200	2.54509876876894\\
300	3.72979970207008\\
400	5.00524677084378\\
500	6.2583756676991\\
600	7.37558849940755\\
700	8.65629484501795\\
800	9.89883098417352\\
900	11.1844759694921\\
1000	11.9902536733535\\
};
\addlegendentry{GC3P};

\addplot [color=mycolor1,dashed, ultra thick,mark=diamond,mark options={solid} , mark size = \msize]
  table[row sep=crcr]{%
100	1.05671192122218\\
200	1.92639042002182\\
300	2.82273296688827\\
400	3.70412734649402\\
500	4.64743145979225\\
600	5.52203159602494\\
700	6.40222487874069\\
800	7.34875423183624\\
900	8.27035605962258\\
1000	8.94993906867189\\
};
\addlegendentry{C3P};

\end{axis}
\end{tikzpicture}%

%% file: NEWR1000_NumWorkvariable_NumColl1o4_iter100.tex
%
%
\definecolor{mycolor1}{rgb}{0.0, 0.58, 0.71}
\definecolor{mycolor3}{rgb}{1.0, 0.65, 0.0}
\definecolor{mycolor4}{rgb}{0.53, 0.66, 0.42} 
\definecolor{mycolor2}{rgb}{0.8, 0.31, 0.36}
\begin{tikzpicture}
\def \msize{3}
\begin{axis}[%
width=0.951\figurewidth,
height=\figureheight,
at={(0\figurewidth,0\figureheight)},
scale only axis,
xmin=10,
xmax=100,
xlabel={\Large Number of workers $n$},
ymin=0,
ymax=300,
ylabel={\Large Average completion time},
axis background/.style={fill=white},
legend style={legend cell align=left,align=left,draw=white!15!black}
]

\addplot [color=mycolor4,dashed,ultra thick,mark=+,mark options={solid}, mark size = \msize]
  table[row sep=crcr]{%
10	170.677050269582\\
20	68.9574108890923\\
30	49.5084691842434\\
40	34.9075275339232\\
50	29.243390261425\\
60	23.5316711066132\\
70	20.8916070054565\\
80	17.811855117509\\
90	16.1808864375402\\
100	14.2833209957022\\
};
\addlegendentry{Staircase codes};

\addplot [color=mycolor3,dashed,ultra thick,mark=o,mark options={solid}, mark size = \msize]
  table[row sep=crcr]{%
10	141.495725512529\\
20	60.33855519074\\
30	42.532494736894\\
40	30.3235143578628\\
50	25.1482853014765\\
60	20.3641547025128\\
70	17.9541819199496\\
80	15.3399688959021\\
90	13.9645830738505\\
100	12.328047010354\\
};
\addlegendentry{PRAC};

\addplot [color=mycolor2,solid,mark=+,mark options={solid}, ultra thick,mark size = \msize]
  table[row sep=crcr]{%
10	123.974345571857\\
20	60.2854318224441\\
30	40.5806420400909\\
40	30.2921877129104\\
50	24.3987502325139\\
60	20.3451717188176\\
70	17.54348071153\\
80	15.3254989750666\\
90	13.6925900881264\\
100	12.3163622628771\\
};
\addlegendentry{GC3P};

\addplot [color=mycolor1,dashed, ultra thick,mark=diamond,mark options={solid} , mark size = \msize]
  table[row sep=crcr]{%
10	58.6549058352498\\
20	26.4386035796407\\
30	18.264572444802\\
40	13.3424547609769\\
50	10.8982644493225\\
60	8.94411631385786\\
70	7.80911737928998\\
80	6.74966360995802\\
90	6.09712336298765\\
100	5.43861822022874\\
};
\addlegendentry{C3P};

\end{axis}
\end{tikzpicture}%

%% file: NEWR1000_NumWorkvariable_NumColl13_iter100.tex
%
%
\definecolor{mycolor1}{rgb}{0.0, 0.58, 0.71}
\definecolor{mycolor3}{rgb}{1.0, 0.65, 0.0}
\definecolor{mycolor4}{rgb}{0.53, 0.66, 0.42} 
\definecolor{mycolor2}{rgb}{0.8, 0.31, 0.36}
\begin{tikzpicture}
\def \msize{3}
\begin{axis}[%
width=0.951\figurewidth,
height=\figureheight,
at={(0\figurewidth,0\figureheight)},
scale only axis,
xmin=20,
xmax=100,
xlabel={\Large Number of workers $n$},
ymin=0,
ymax=300,
ylabel={\Large Average completion time},
axis background/.style={fill=white},
legend style={legend cell align=left,align=left,draw=white!15!black}
]

\addplot [color=mycolor4,dashed,ultra thick,mark=+,mark options={solid}, mark size = \msize]
  table[row sep=crcr]{%
20	291.605753680434\\
30	76.6985423437678\\
40	40.9327141715677\\
50	29.2349015081406\\
60	22.0982404988382\\
70	18.2053183723176\\
80	15.1703561124333\\
90	13.2339099062179\\
100	11.5837097489838\\
};
\addlegendentry{Staircase codes};

\addplot [color=mycolor3,dashed,ultra thick,mark=o,mark options={solid}, mark size = \msize]
  table[row sep=crcr]{%
20	195.678902923971\\
30	61.9088424583411\\
40	35.3694700517878\\
50	25.1929314947661\\
60	17.5967700889406\\
70	13.7354194710647\\
80	10.8674947787487\\
90	9.28392739994736\\
100	7.86201314885801\\
};
\addlegendentry{PRAC};

\addplot [color=mycolor2,solid,mark=+,mark options={solid}, ultra thick,mark size = \msize]
  table[row sep=crcr]{%
20	100.185354668941\\
30	44.8924934187705\\
40	31.7185222018174\\
50	24.4521627613519\\
60	17.3574883377777\\
70	13.4796639009596\\
80	10.6153788440569\\
90	9.04229229832404\\
100	7.6410587612594\\
};
\addlegendentry{GC3P};

\addplot [color=mycolor1,dashed, ultra thick,mark=diamond,mark options={solid} , mark size = \msize]
  table[row sep=crcr]{%
20	26.4580681895177\\
30	18.3235408606476\\
40	13.3279559204105\\
50	10.9141854379201\\
60	8.92458131855267\\
70	7.80476539355301\\
80	6.75122735102084\\
90	6.09344651814085\\
100	5.42879727673192\\
};
\addlegendentry{C3P};

\end{axis}
\end{tikzpicture}%

%% file: NEWR1000_NumWor50_NumCollvariable_iter100.tex
%
%
\definecolor{mycolor1}{rgb}{0.0, 0.58, 0.71}
\definecolor{mycolor3}{rgb}{1.0, 0.65, 0.0}
\definecolor{mycolor4}{rgb}{0.53, 0.66, 0.42} 
\definecolor{mycolor2}{rgb}{0.8, 0.31, 0.36}
\begin{tikzpicture}
\def \msize{2}
\begin{axis}[%
width=0.951\figurewidth,
height=\figureheight,
at={(0\figurewidth,0\figureheight)},
scale only axis,
xmin=0,
xmax=40,
xlabel={\Large Number of colluding workers $z$},
ymin=0,
ymax=250,
ylabel={\Large Average completion time},
axis background/.style={fill=white},
legend style={legend cell align=left,align=left,draw=white!15!black, at={(0.5,1)}}
]

\addplot [color=mycolor4,dashed,ultra thick,mark=+,mark options={solid}, mark size = \msize]
  table[row sep=crcr]{%
1	19.6119583831531\\
2	20.2071936020707\\
3	20.7649469611366\\
4	21.3839556206288\\
5	22.0920147631752\\
6	22.713058534048\\
7	23.5104989042457\\
8	24.2987016215011\\
9	25.1511058701219\\
10	26.049504796285\\
11	27.0334089636711\\
12	28.1147222549489\\
13	29.2149190240437\\
14	30.5027128285516\\
15	31.8562393985089\\
16	33.3247196047328\\
17	34.9317518019799\\
18	36.7546569816453\\
19	38.8196985978089\\
20	41.0032723855897\\
21	43.5447612869781\\
22	46.4244535549435\\
23	49.6910386106366\\
24	53.5411642262491\\
25	57.9152042557586\\
26	63.1041798139344\\
27	69.268796570276\\
28	76.9126775443301\\
29	86.5873356523736\\
30	98.9421319582002\\
31	108.635055472889\\
32	114.691839645638\\
33	121.293717161376\\
34	128.919572235846\\
35	137.354568892951\\
36	147.304837342042\\
37	158.470812230969\\
38	171.567616358151\\
39	187.339166990817\\
40	206.134914696513\\
};
\addlegendentry{Staircase codes};

\addplot [color=mycolor3,dashed,ultra thick,mark=o,mark options={solid}, mark size = \msize]
  table[row sep=crcr]{%
1	11.5123440416478\\
2	12.134067671669\\
3	12.7924361185052\\
4	13.5487081488437\\
5	14.3619049359516\\
6	15.322786459536\\
7	16.35563384221\\
8	17.5143485198141\\
9	18.7970829832383\\
10	20.3422104304155\\
11	22.208157496341\\
12	24.1779103155688\\
13	25.1136099323211\\
14	26.220011391931\\
15	27.3379183797101\\
16	28.5087226837669\\
17	29.8233596365006\\
18	31.1597611068351\\
19	32.4891520674162\\
20	34.205291876436\\
21	35.9222770232232\\
22	37.8486734557468\\
23	39.9820060281678\\
24	42.3236638490017\\
25	44.9842366132606\\
26	47.7930034683581\\
27	51.2192327709451\\
28	54.9894072512652\\
29	59.4095119440541\\
30	64.5078778195951\\
31	70.9086397653581\\
32	78.2296758690923\\
33	86.9877392143233\\
34	98.7542996945167\\
35	113.838488438833\\
36	133.78595124374\\
37	162.503099442179\\
38	177.237368408903\\
39	194.376276334002\\
40	214.435166982401\\
};
\addlegendentry{PRAC};

\addplot [color=mycolor2,solid,mark=+,mark options={solid}, ultra thick,mark size = \msize]
  table[row sep=crcr]{%
1	11.4380804309719\\
2	12.0083870039201\\
3	12.6207921318584\\
4	13.335333845487\\
5	14.1171691702393\\
6	15.024719826986\\
7	16.0374986831442\\
8	17.1961189634434\\
9	18.5087215507409\\
10	20.1015517919586\\
11	22.0261074458003\\
12	24.1557469480789\\
13	24.3815180543015\\
14	24.7096658978276\\
15	25.025781097577\\
16	25.3103899207915\\
17	25.6482993595174\\
18	25.9095405008917\\
19	26.1270387584603\\
20	26.5175088662326\\
21	26.8573469161454\\
22	27.1692506509859\\
23	27.5483272683636\\
24	27.8694234293205\\
25	28.3124727915325\\
26	29.3880662218185\\
27	30.7442062248794\\
28	32.0573824092752\\
29	33.6011641290935\\
30	35.2289665153075\\
31	37.0111988452392\\
32	39.1770127705424\\
33	41.3525248006505\\
34	43.987437059451\\
35	46.9788747086384\\
36	50.1329893542992\\
37	53.9896774722919\\
38	58.5708975224675\\
39	63.8473095193331\\
40	70.2911306700171\\
};
\addlegendentry{GC3P};

\addplot [color=mycolor1,dashed, ultra thick,mark=diamond,mark options={solid} , mark size = \msize]
  table[row sep=crcr]{%
1	10.9149665232018\\
2	10.914770469746\\
3	10.8938378940697\\
4	10.9108462086372\\
5	10.8944478688514\\
6	10.9243501947861\\
7	10.9236951063585\\
8	10.9322160787192\\
9	10.9034146921505\\
10	10.8946300381667\\
11	10.9232326861857\\
12	10.8849305497382\\
13	10.8888763667838\\
14	10.9042624796751\\
15	10.8927818382706\\
16	10.9273721063881\\
17	10.911363641579\\
18	10.8921301807534\\
19	10.8698418526278\\
20	10.9224639797535\\
21	10.9077996250002\\
22	10.9272419133915\\
23	10.9134044412759\\
24	10.9222166024485\\
25	10.9221812829894\\
26	10.8968030891548\\
27	10.913147556179\\
28	10.8974165176403\\
29	10.8987477712447\\
30	10.8773126426725\\
31	10.9013128726641\\
32	10.8915028859702\\
33	10.9133736398001\\
34	10.909387846508\\
35	10.9071523536621\\
36	10.9247114141559\\
37	10.8918285313873\\
38	10.9045963154831\\
39	10.8975802946067\\
40	10.9043485684963\\
};
\addlegendentry{C3P};

\end{axis}
\end{tikzpicture}%

%% file: Slowest_N_Z_Homo_NumCollvariable_NumWrok50_iter100.tex
%
%
\definecolor{mycolor1}{rgb}{0.0, 0.58, 0.71}
\definecolor{mycolor3}{rgb}{1.0, 0.65, 0.0}
\definecolor{mycolor4}{rgb}{0.53, 0.66, 0.42} 
\definecolor{mycolor2}{rgb}{0.8, 0.31, 0.36}
\begin{tikzpicture}
\def \msize{2}
\begin{axis}[%
width=0.951\figurewidth,
height=\figureheight,
at={(0\figurewidth,0\figureheight)},
scale only axis,
xmin=0,
xmax=40,
xlabel={\Large Number of colluding workers $z$},
ymin=0,
ymax=250,
ylabel={\Large Average completion time},
axis background/.style={fill=white},
legend style={legend cell align=left,align=left,draw=white!15!black, at={(0.5,1)}}
]

\addplot [color=mycolor3,dashed,ultra thick,mark=o,mark options={solid}, mark size = \msize]
  table[row sep=crcr]{%
1	43.661622332498\\
2	44.4858195881264\\
3	45.6033763706745\\
4	46.4043688168998\\
5	47.343982946331\\
6	48.4250993838749\\
7	49.713649643721\\
8	50.7911541967181\\
9	52.0071015402494\\
10	53.1999237384814\\
11	54.5230501222176\\
12	56.0690859827673\\
13	57.6022622745391\\
14	59.1433390822789\\
15	60.8441803275899\\
16	62.5154007255706\\
17	64.3779602776508\\
18	66.4464059532054\\
19	68.5889878490111\\
20	70.7219451861782\\
21	73.1576700411972\\
22	75.8861932108574\\
23	78.581143828442\\
24	81.6618790595288\\
25	85.0307758168076\\
26	88.3621667235451\\
27	92.1587862216923\\
28	96.0823987859343\\
29	100.916806703864\\
30	106.059884734502\\
31	111.453892066583\\
32	117.564502568793\\
33	124.082957524222\\
34	131.823858077595\\
35	140.627901706189\\
36	150.838547866979\\
37	162.583236594476\\
38	176.230597199247\\
39	191.678906671738\\
40	210.552008461574\\
};
\addlegendentry{PRAC};

\addplot [color=mycolor4,dashed,ultra thick,mark=+,mark options={solid}, mark size = \msize]
  table[row sep=crcr]{%
1	42.8891592277575\\
2	43.760973538901\\
3	44.6701159839236\\
4	45.5706155269428\\
5	46.6504575813439\\
6	47.6229556274356\\
7	48.7673326120641\\
8	49.8312436552887\\
9	51.0740482158038\\
10	52.3018960786455\\
11	53.6384488468755\\
12	54.9431418821782\\
13	56.4690865875565\\
14	57.9325597240736\\
15	59.5772829158251\\
16	61.3905693249365\\
17	63.2108100633284\\
18	65.0929315583829\\
19	67.15907521029\\
20	69.4506825586277\\
21	71.832694344417\\
22	74.2722295037711\\
23	77.0059884437662\\
24	79.9424077891538\\
25	83.1643485890522\\
26	86.595947677293\\
27	90.1384651115445\\
28	94.1942668490218\\
29	98.6554887336575\\
30	103.525290562122\\
31	108.987265729367\\
32	114.777916944716\\
33	121.808930101952\\
34	129.022723978235\\
35	137.682957719801\\
36	147.735480948989\\
37	158.370920704099\\
38	171.40375587708\\
39	186.885129260243\\
40	205.169403451901\\
};
\addlegendentry{Staircase codes};


\addplot [color=mycolor2,solid,mark=+,mark options={solid}, ultra thick,mark size = \msize]
  table[row sep=crcr]{%
1	43.6204147347081\\
2	44.4474680420105\\
3	45.5570473259467\\
4	46.3658572552558\\
5	47.3034263884088\\
6	48.3738032774749\\
7	49.6614669275534\\
8	50.7460353038849\\
9	51.9618343957285\\
10	53.1540581020171\\
11	54.4738428447677\\
12	56.0190889602264\\
13	57.5523715803367\\
14	59.0858910394278\\
15	60.7914110883362\\
16	62.4618712810789\\
17	64.3066441579972\\
18	66.3907717260393\\
19	68.5042099377229\\
20	70.6463101439186\\
21	73.0930867605701\\
22	75.824045391208\\
23	78.5091628063399\\
24	81.5893505609526\\
25	84.9539384549053\\
26	88.2762946405023\\
27	92.06953284083\\
28	95.9800654741611\\
29	100.824691784572\\
30	105.952960902446\\
31	111.330877012838\\
32	117.460202218843\\
33	123.964308199931\\
34	131.699559343877\\
35	140.479191892306\\
36	150.681616604802\\
37	162.442172786131\\
38	176.075705452059\\
39	191.492667073899\\
40	210.348536485892\\
};
\addlegendentry{GC3P};

\addplot [color=mycolor1,dashed, ultra thick,mark=diamond,mark options={solid} , mark size = \msize]
  table[row sep=crcr]{%
1	41.1465920576876\\
2	35.6284343025092\\
3	34.5902337787458\\
4	30.5509678360195\\
5	29.6150143346085\\
6	26.6405589853927\\
7	26.1141925179152\\
8	23.7452436846346\\
9	23.2048539494944\\
10	21.3542705111992\\
11	20.978158835216\\
12	19.4375090889155\\
13	19.1139514722956\\
14	17.7960892832344\\
15	17.5466446297045\\
16	16.4179277522727\\
17	16.1811849770095\\
18	15.278496779641\\
19	15.0807958672124\\
20	14.2575591839599\\
21	14.0642237611258\\
22	13.3456432982461\\
23	13.1727594681083\\
24	12.6102281721487\\
25	12.4245157616402\\
26	11.8763218635243\\
27	11.7274846114474\\
28	11.2405105975276\\
29	11.1521967923815\\
30	10.7194378513984\\
31	10.6008625411038\\
32	10.1515236522223\\
33	10.1068972371975\\
34	9.69846410928152\\
35	9.61313542872636\\
36	9.28105803851345\\
37	9.20472615796884\\
38	8.91325375535446\\
39	8.82084005723514\\
40	8.56132787418166\\
};
\addlegendentry{C3P};

\end{axis}
\end{tikzpicture}%